\makeatletter
\let\latexaddtocontents\addtocontents 
\makeatother

\documentclass[12pt]{amsart}

\makeatletter
\let\addtocontents\latexaddtocontents 
\makeatother

\pdfoutput=1
\usepackage[utf8]{inputenc}
\usepackage[english]{babel}
\usepackage[T1]{fontenc}
\usepackage[margin=1in]{geometry}
\usepackage{amsmath,amssymb,amsthm,amsfonts,bm,bbm,braket,diagbox,graphicx,tabularx,hhline,longtable,mathtools,mathrsfs,tikz,tikz-cd,upref,xurl,enumitem,physics,multirow,caption,float,lmodern,rotating,ragged2e,adjustbox,xcolor,colortbl,comment,scalerel,fix-cm,etoolbox,subfig,xstring,booktabs,array,fancyhdr,orcidlink}

\usepackage[normalem]{ulem}
\usepackage[
    backend=biber,
    style=numeric-comp,
    sorting=none,
    maxbibnames=2,
    minbibnames=1,
    maxcitenames=2,
    mincitenames=1,
    giveninits=true,
    terseinits=true
]{biblatex}
\AtEveryBibitem{
\clearfield{issn}
\clearfield{doi}
\clearfield{url}
}
\renewbibmacro*{volume+number+eid}{%
\printfield{volume}%
\printfield[parens]{number}%
\printfield{eid}}
\usepackage{hyperref}
\newtheorem{theorem}{Theorem}
\newtheorem{prop}[theorem]{Proposition}
\newtheorem{cor}[theorem]{Corollary}

\newtheorem*{prop*}{Proposition}
\DeclareMathOperator*{\argmax}{arg\,max}
\DeclareMathOperator*{\argmin}{arg\,min}
\newcommand{\UNOmath}{Department of Mathematics, University of New Orleans, New Orleans, LA 70148, USA}

\usepackage{xcolor}

\newcommand{\affmark}[1]{\textsuperscript{#1}}

\begin{document}

\title{No-signalling-projection-invariant Bell inequalities}

\author[S. Patra]{Soumyadip Patra\affmark{1,*}}
\thanks{\textsuperscript{*}Corresponding author: \texttt{spatra@uno.edu}.}
\author[J. Prakash]{Jitendra Prakash\affmark{1,2}}
\author[A. Paudel]{Aaditya Paudel\affmark{1,3}}
\author[P. Bierhorst]{Peter Bierhorst\affmark{1}}

\makeatletter
\global\let\addresses\@empty
\makeatother

\address[1]{\UNOmath}
\address[2]{Jindal Centre for Digital Sciences, O.~P.~Jindal Global University,
Sonipat, Haryana 131001, India}
\address[3]{Department of Mathematics, Miami University, Oxford, OH 45056, USA}

\begin{abstract}
In this paper, we highlight how any Bell inequality for a configuration involving $n$ parties each performing one of $m$ binary-outcome measurements has a canonical form that is no-signalling-projection invariant. Specifically, the $L^2$-projection of weakly signalling data onto the no-signalling polytope leaves the violation of this canonical Bell inequality unchanged. Our methods allow us to derive a general closed formula for the projection and present a substantially more computationally simple procedure for its evaluation. We also show this can be generalised to non-standard projections of potential interest for certain applications. No-signalling projections serve as a preliminary step before undertaking any device-independent application involving Bell experiment data, such as hypothesis testing against local realism, random number generation and entanglement detection.
\end{abstract}

\maketitle

\section{Introduction}
The theoretical treatment of Bell nonlocality and its implications can be approached through a study of sets of conditional distributions $p(\mathbf{a}\lvert\mathbf{x})$ of outcomes given measurement settings, which describe the observable behaviour of a Bell experiment in the asymptotic limit of infinitely many rounds. However, data obtained from real Bell experiments result from finitely many trials and have statistical fluctuations, and consequently, an empirical estimate of $p(\mathbf{a}\lvert\mathbf{x})$ obtained from the data is almost always guaranteed to have \textit{weak} signalling effects, in the sense of violating the \textit{no-signalling} conditions which assert that a party's marginal probability of outcome must be independent of the other party's choice of measurement setting. Table~\ref{tab:Distr_Sig_222} shows an example of an experimental count data from a two-party, two-setting, two-outcome Bell test demonstrating the kind of weak-signalling effects that can appear in experimental data.

In applications of Bell experiments---such as certifiable device-independent quantum randomness~\cite{Acín2016,Pironio2010}, quantum key distribution~\cite{Masanes_2011,Zapatero_2023}, entanglement certification~\cite{Moroder_2013}, and quantification of statistical nonlocality in multipartite scenarios~\cite{PB2024}---it is necessary to obtain a point estimate of the true trial distribution before proceeding with the implementation of the protocol to ensure optimal performance. Deriving the frequency from some preliminary raw count data and using it as a point estimate for a trial distribution without accounting for weak signalling effects can lead to misleading interpretations. For instance, given such an estimate (the empirical distribution) obtained from the unprocessed experimental data, and a Bell inequality describing a facet of the local polytope, one can find a form of the inequality that is satisfied by the behaviour and a form that is violated by it~\cite{Scarani2019} (see Figure 2.2 therein); this is because the empirical distribution is weakly-signalling, and the no-signalling terms used to re-express the inequality no longer vanish, so the resulting ``equivalent'' form can flip from satisfaction to violation.

Instead of using the empirical frequency derived from the count data as a point estimate of the true trial distribution, one should find its ``closest'' approximation satisfying the no-signalling conditions before using it in a protocol which involves the violation of Bell inequalities. This is because in principle, the target behaviour in a Bell experiment obeys the no-signalling constraints. Indeed, implementations that enforce space-like separations of each party's measurement choice from the other party's recorded outcome (closing the locality loophole) preclude superluminal communication, and quantum mechanics respects special relativity and therefore meets the no-signalling requirements. Note that without space-like separations, apparent violations of no-signalling can also arise from experimental sub-optimalities, e.g., slow drifts across sequential fixed-setting blocks, or imperfect randomisation~\cite{Smania25, Bednorz17,Rybotycki25,Liang19,Soltan20}. In such cases, the observed signalling is best interpreted as a deviation from an ideal, drift-free experiment one aims to estimate; the signalling-contaminated empirical distribution is an unphysical estimate of this underlying no-signalling behaviour.

\begin{table}[!htbp]
\caption{\label{tab:Distr_Sig_222}Count data from~\cite{Bierhorst2018} for a $(2,2,2)$ configuration Bell experiment, a bipartite Bell test where each party has a choice between two dichotomic measurements. Each entry $N(ab\lvert xy)$ represents the number of trials yielding outcomes $a$ and $b$ given measurement settings $x$ and $y$. It can be checked that the empirical distribution exhibits weak-signalling effects. For instance, the sum of the entries marked \textsuperscript{\dag} divided by their row sum, $(N(00\lvert 00)+N(01\lvert 00))/\sum_{a,b}N(ab\lvert 00)\approx 4.012\times 10^{-3}$, does not equal the sum of the entries marked \textsuperscript{\ddag} divided by their row sum, $(N(00\lvert 01)+N(01\lvert 01))/\sum_{a,b}N(ab\lvert 01)\approx 3.983\times 10^{-3}$. In other words, the marginal probability of outcome $a=0$ for party $\mathsf A$ varies with the choice of input $y$ of party $\mathsf B$. Ideally, this marginal should be independent of the choice of setting $y$ for $\mathsf B$, as required by the no-signalling condition.}
\centering
\begin{tabular}{rrcccc}
\hline\hline
 &  & \multicolumn{4}{c}{$ab$}\\
\cline{2-6}
 &  & $00$ & $01$ & $10$ & $11$ \\
\cline{2-6}
\multirow{4}{*}{\rotatebox[origin=c]{90}{$xy$}}
 & $00$ & 3166\textsuperscript{\dag} & 1851\textsuperscript{\dag} & 2043 & 1243520 \\
 & $01$ & 3637\textsuperscript{\ddag} & 1338\textsuperscript{\ddag} & 13544 & 1230633 \\
 & $10$ & 3992 & 13752 & 1226 & 1230686 \\
 & $11$ & 357 & 17648 & 16841 & 1215766 \\
\hline\hline
\end{tabular}
\end{table}

Two possible methods to achieve a no-signalling approximation to weakly signalling empirical data are based on the procedure of (a) maximum-likelihood (ML) and (b) $L^{2}$-projection onto the smallest-dimensional affine subspace containing the no-signalling polytope. Given a count dataset, in the ML procedure, one obtains the empirical frequency from the raw data, and then maximises its likelihood with respect to a trial distribution whose corresponding settings-conditional outcome distribution is restricted to satisfy the normalisation and the no-signalling conditions. This preprocessing has been used in the prediction-based ratio approach to hypothesis testing, for instance, in~\cite{PB2024,Bierhorst2018,Zhang20LowLatencyDIQRNG,ZhangGlancyKnill2011,ZhangGlancyKnill2013}, and in device-independent randomness generation via the probability estimation framework~\cite{ZhangKnillBierhorst2018,KZB2020,PatraBierhorst2023,Zhang20QEF,Zhang2021}. While the ML estimator is a statistically natural choice for removing finite statistics fluctuations like those in Table \ref{tab:Distr_Sig_222}, it typically lacks a closed-form expression; in practice one solves an iterative convex program enforcing linear equalities (no-signalling and normalisation) and non-negativity inequalities, which will grow in complexity with the Bell scenario as the number of parties, settings, and outcomes increases. Furthermore, the naturality of the ML procedure is much less clear for removing signalling artifacts caused by experimental sub-optimalities such as slow drifts over the course of an experiment in which fixed measurement settings are implemented for a large block of repeated trials, with different measurement settings for different blocks of trials. An estimation procedure akin to averaging across blocks is more appropriate in this case; projection-based denoising can be used to temper this form of weak signalling (though this should not substitute for mitigating experiment-specific causes, if possible).

This paper focuses on the $L^2$-projection method which yields the nearest no-signalling approximation of an empirical frequency vector $\mathbf f\coloneqq\{f(\mathbf{a}\lvert\mathbf{x})\}$, according to the Euclidean metric, by orthogonally projecting it onto $\mathcal A$, the affine hull of the no-signalling set $\mathcal{P}_{\mathrm{NS}}$ of behaviours; it has been previously explored in~\cite{Bancal_2014, Lin2018}. A procedure for performing this affine projection is to first project $\mathbf f$ onto the kernel of the equality-constraint matrix $A_{\mathrm{eq}}$ for the linear system $A_{\mathrm{eq}}\mathbf{p}=\mathbf{b}$ encoding the equality conditions of no-signalling and normalisation, followed by a translation by a vector belonging to $\mathcal{P}_{\mathrm{NS}}$. We prove in Section~\ref{s:CorrsFornm2BellScenario} (Proposition~\ref{prop:CorrVecBelongsToNullSpace}) with an elementary argument that in the $(n,m,2)$ Bell scenario an $L^2$-projection of an empirical frequency onto $\mathcal A$ preserves the values of full correlator terms as well as a uniformly-averaged form of marginal correlator terms. This extends an implicit observation in Appendix D of~\cite{Bancal_2014} for the basic (2,2,2) case, described more explicitly in Appendix B.3 of~\cite{Lin2018} where the authors indicate the result can be generalised to higher order scenarios employing the machinery of symmetries under local transformations~\cite{Renou_2017,Rosset2020}. In contrast to the technical approach of~\cite{Renou_2017,Rosset2020}, which enlists representation theory to study these symmetry groups, our direct approach here uses elementary methods while illuminating the close relationship of the result to the sparse form of the linear no-signalling and normalisation constraints. The invariance of the uniformly-averaged marginal correlators then allows us to explicitly present a simplified closed-form for the projection, avoiding a direct (more computationally costly) projector construction---i.e., computing a basis $\{\mathbf{b}_i\}_{i=1}^{k}$ of the kernel of $A_{\mathrm{eq}}$ and forming the projector $B(B^{T}B)^{-1}B^T$ with $B=\begin{pmatrix}
    \mathbf{b}_{1} & \mathbf{b}_{2} & \cdots & \mathbf{b}_{k}
\end{pmatrix}$, a route requiring repeated elimination to obtain $B$ and solving with the dense Gram matrix $B^{T}B$, which becomes much more difficult in higher order Bell scenarios. We are also able to generalise the result to weighted $L^2$ projections, which are of interest in removing signalling artifacts in applications where measurement settings are sampled non-uniformly, such as certain protocols for device-independent quantum key distribution~\cite{Acín_2006,Arnon-Friedman2018,Schwonnek2021} and device-independent random number expansion~\cite{Miller2017,Bhavsar_2023,Shalm2021,Liu2021}.

The above findings about the projection method lead to an important consequence that we highlight. Because uniformly-averaged marginal correlators (together with the full correlators) form a canonical spanning family for linear functionals on behaviours, any Bell expression can be rewritten in terms of these correlators. Consequently, Bell inequalities in this canonical correlator form (CHSH included as a special case) are termwise invariant under the projection, and hence so is the entire expression. Therefore, the no-signalling projected estimate attains exactly the same Bell value (and thus the same degree of violation) as that of the raw empirical distribution. The projection-invariant Bell expression can then be taken as its canonical form. Canonical forms for Bell expressions can also be derived using the framework introduced in Section 6.1.1 of \cite{Rosset2020}. However, this approach involves more elaborate machinery: one must first decompose an appropriate linear space associated with the parties into the normalisation, ``allowed'' and ``forbidden'' subspaces, and then project the Bell functional onto the ``allowed'' subspace. Our derivation here highlights the strong connection to correlator form, providing an immediate method for Bell expressions already expressed in terms of correlators, and a simple recipe for general Bell expressions. As every Bell inequality in the $(n,m,2)$ configuration admits such a representation, experimental analyses can utilise this canonical formulation as a starting point that minimises leverage from signalling---whether arising from unavoidable finite-sample fluctuations (even in loophole-free implementations) or from mild, experiment-specific systematics (e.g., drifts, blockwise settings, imperfect randomisation). The canonical form isolates the violation attributable to strictly no-signalling nonlocality, providing a cleaner and more comparable metric of nonlocal behaviour across experiments, and will serve as a useful baseline for choosing among otherwise equivalent representations of Bell expressions across a wide range of applications. 

\section{Projection in the \texorpdfstring{$(n,m,2)$}{nm2} Bell scenario}\label{s:Section2}
\subsection{The affine subspace \texorpdfstring{$\mathcal{A}\supset\mathcal{P}_{\mathrm{NS}}$}{AcontainsN}}
We now present the mathematical setting in which we derive our results. The $(n,m,2)$ Bell scenario is viewed as a system consisting of $n$ non-communicating parties $\mathsf{A}_{1},\mathsf{A}_{2},\ldots,\mathsf{A}_{n}$ that share a quantum-entangled resource (e.g., an $n$-partite quantum state $\ket{\Psi}_{\mathrm{GHZ}}=(\ket{0}^{\otimes n}+\ket{1}^{\otimes n})/\sqrt{2}$). Throughout the paper, for $n\in\mathbb{N}$, we use the shorthand notation $[n]$ for the set $\{1,2,\ldots,n\}$. For $i\in[n]$, party $\mathsf{A}_{i}$ provides an input $x_{i}$ chosen from the $m$-element set $\mathcal{S}\coloneqq\{\mathsf{s}_{1},\mathsf{s}_{2},\ldots,\mathsf{s}_{m}\}$ to her share of the resource and obtains a binary outcome $a_{i}\in\mathcal{O}\coloneqq\{0,1\}$. While larger outcome sets can be considered in more general Bell scenarios, this paper focuses on binary outcomes where they have the most natural definition. An $(n,m,2)$ behaviour $\mathbf{p}$ is fully specified by the conditional probabilities $p(\mathbf{a}\lvert\mathbf{x})$, with $\mathbf{a}=(a_{1},a_{2},\ldots,a_{n})$ and $\mathbf{x}=(x_{1},x_{2},\ldots,x_{n})$: 
\begin{equation}
\mathbf{p}\coloneqq\left\lbrace p(\mathbf{a}\lvert\mathbf{x})\colon(\mathbf{a},\mathbf{x})\in\mathcal{O}^{n}\times\mathcal{S}^{n} \right\rbrace.
\end{equation} 
Geometrically, a behaviour $\mathbf{p}$ can be viewed as a vector belonging to the ambient space $\mathbb{R}^{d}$, with $d=(2m)^n$, whose components are the settings-conditional outcome probabilities $p(\mathbf{a}\lvert\mathbf{x})$. Being a vector of conditional probabilities, a behaviour belongs to the orthant $\mathbb{R}_{+}^d$ of vectors with non-negative entries and satisfies the normalisation conditions:
\begin{equation}\label{eq:Normalisation_nk2}
\sum_{\mathbf{a}}p(\mathbf{a}\lvert\mathbf{x})=1,\,\forall\mathbf{x}\in\mathcal{S}^{n},    
\end{equation} 
If, in addition to the non-negativity and normalisation conditions, a behaviour also satisfies the condition that the probabilities of outcomes observed by any subset of parties are not influenced by the choice of inputs of the parties in the complement set, then we refer to it as an $(n,m,2)$ no-signalling behaviour. This can be expressed algebraically through linear equalities in~\eqref{eq:NoSignm2} which are mathematical statements of the condition that for $i\in[n]$ the input $x_{i}\in\mathcal{S}$ chosen by party $\mathsf{A}_{i}$ does not influence the outcomes of the remaining $n-1$ parties.
\begin{equation}\label{eq:NoSignm2}
p(\mathbf{a}_{\neg i}\lvert\mathbf{x}_{\neg i},x_i=\mathsf{s}_1)=p(\mathbf{a}_{\neg i}\lvert\mathbf{x}_{\neg i},x_i=\mathsf{s}_r), \forall i\in[n],r\in[m]\setminus\{1\},(\mathbf{a}_{\neg i},\mathbf{x}_{\neg i}).
\end{equation}
Here, $\mathbf{a}_{\neg i}$ is the outcome combinations of all parties except party $\mathsf{A}_i$, $\mathbf{x}_{\neg i}$ is the measurement-settings combination for the same parties, and $p(\mathbf{a}_{\neg i}\lvert\mathbf{x}_{\neg i},x_i = \mathsf{s}_r)$ (obtained by summing $p(\mathbf{a}\lvert\mathbf{x}_{\neg i},x_{i}=\mathsf{s}_r)$ over $a_{i}$) is the marginal probability of observing outcomes $\mathbf{a}_{\neg i}$ given measurement settings $\mathbf{x}_{\neg i}$ when party $\mathsf{A}_{i}$ chooses $\mathsf{s}_r$ (i.e., $x_i = \mathsf{s}_r$). We need not include instances such as:
\begin{equation*}
p(\mathbf{a}_{\neg i}\lvert\mathbf{x}_{\neg i},x_i=\mathsf{s}_2)=p(\mathbf{a}_{\neg i}\lvert\mathbf{x}_{\neg i},x_i=\mathsf{s}_r), \forall i\in[n],r\in[m]\setminus\{2\},(\mathbf{a}_{\neg i},\mathbf{x}_{\neg i}),
\end{equation*}
as they are implied by the equality conditions in~\eqref{eq:NoSignm2}. A straightforward counting then reveals that for party $i$ there are $(m-1)(2m)^{n-1}$ such equality conditions leading to a total of $n(m-1)(2m)^{n-1}$ no-signalling conditions (when we consider them for all $i\in[n]$). These conditions completely specify \textit{all possible} no-signalling conditions for the $(n,m,2)$ Bell scenario, in the sense that they can be used to derive the condition that the inputs chosen by any $k$-party subset $\{\mathsf{A}_{i_1},\mathsf{A}_{i_2},\ldots,\mathsf{A}_{I}\}$, where $1\le i_{1}<i_{2}<\cdots<I\le n$ and $k\in[n-1]$, do not influence the outcomes of the remaining $n-k$ parties (constituting the complement of that $k$-party subset)~\cite{Bierhorst2024,Pironio2004PhD}. The $n(m-1)(2m)^{n-1}$ no-signalling conditions in~\eqref{eq:NoSignm2} and the $m^{n}$ normalisation conditions in~\eqref{eq:Normalisation_nk2} together comprise the equality conditions for specifying a no-signalling behaviour and can be represented as $A_{\mathrm{eq}}\mathbf{p}=\mathbf{b}$, where $A_{\mathrm{eq}}\in\mathbb{R}^{t\times d}$ with $t=m^n + n(m-1)(2m)^{n-1}$ and $d=(2m)^n$. Each row of $A_{\mathrm{eq}}$ is a vector $\mathbf{m}_{r}\in\mathbb{R}^d$ whose components are the coefficients of the components of the probability vector $\mathbf{p}$ for an equality condition. The first $n(m-1)(2m)^{n-1}$ entries of the vector $\mathbf{b}\in\mathbb{R}^{t}$ are $0$ (corresponding to no-signalling conditions) and the remaining $m^{n}$ entries are $1$ (corresponding to normalisation conditions). Note some of the no-signalling conditions are redundant as they can be derived using other no-signalling conditions and normalisation conditions,which implies that the rows of $A_{\mathrm{eq}}$ corresponding to those redundant no-signalling conditions are linearly dependent on the other ones. This implies that $\mathrm{rank}(A_{\mathrm{eq}})<t$.

The linear system $A_{\mathrm{eq}}\mathbf{p}=\mathbf{b}$ collects only the linear equalities but imposes no inequalities, hence $\mathcal{A}$ is the affine hull of $\mathcal{P}_{\mathrm{NS}}$. Imposing non-negativity (because the entries of $\mathbf{p}$ are probabilities), one recovers the no-signalling polytope, i.e., $\mathcal{P}_{\mathrm{NS}}=\mathcal{A}\cap\mathbb{R}_{+}^{d}$. In the next section, we compare two approaches to obtaining a no-signalling approximation for a weakly-signalling empirical behaviour $\mathbf{f}\coloneqq\{f(\mathbf{a}\lvert\mathbf{x})\}$: the maximum-likelihood approach, and the $L^2$-projection onto $\mathcal A$.

\subsection{Two approaches to no-signalling approximation: ML and \texorpdfstring{$L^2$}{L2}-projection}\label{s:MLvsL2}

While this paper focusses on the $L^2$-projection onto $\mathcal{A}\supset\mathcal{P}_{\mathrm{NS}}$, we first briefly review the maximum-likelihood method for comparison. This context highlights the $L^2$-projection's key advantages: closed-form affine projection and straightforward geometric diagnostics, making it faster, stabler, and more interpretable for our analysis.

For an $(n,m,2)$ Bell experiment comprising multiple experimental trials, let $N(\mathbf{a},\mathbf{x})$ denote the number of trials yielding outcome combination $\mathbf{a}$ when the settings configuration is $\mathbf{x}$, and let $N(\mathbf x)=\sum_{\mathbf{a}}N(\mathbf{a},\mathbf{x})$ denote the total number of trials implementing settings combination $\mathbf{x}$. We denote the observed settings distribution as $\pi(\mathbf{x})=N(\mathbf{x})/N$, where $N=\sum_{\mathbf{a},\mathbf{x}}N(\mathbf{a},\mathbf{x})$. Given a candidate (true) model $\mathbf{p}\in\mathcal{P}_{\mathrm{NS}}$, the likelihood of the observed data is $\prod_{\mathbf{a},\mathbf{x}}p(\mathbf{a}\lvert\mathbf{x})^{N(\mathbf{a},\mathbf{x})}$; taking the (base-2) logarithm and dividing by $N$ one gets $\frac{1}{N}\log_{2}\Big(\prod_{\mathbf{a},\mathbf{x}}p(\mathbf{a}\lvert\mathbf{x})^{N(\mathbf{a},\mathbf{x})}\Big)$ which is equal to
\begin{equation*}
\sum_{\mathbf{a},\mathbf{x}}\pi(\mathbf{x})f(\mathbf{a}\lvert\mathbf{x})\log_{2}p(\mathbf{a}\lvert\mathbf{x}).
\end{equation*}
Maximising this over the no-signalling set $\mathcal{P}_{\mathrm{NS}}$ one gets the ML no-signalling estimate:
\begin{equation}\label{eq:ML_optimisation}
\widehat{\mathbf p}_{\mathrm{ML}}=\argmax_{\mathbf{p}\in\mathcal{P}_{\mathrm{NS}}}\sum_{\mathbf{a},\mathbf{x}}\pi(\mathbf{x})f(\mathbf{a}\lvert\mathbf{x})\log_{2}p(\mathbf{a}\lvert\mathbf{x}).
\end{equation}
The maximisation in~\eqref{eq:ML_optimisation} can be shown to be equivalent to the minimisation of a settings-weighted Kullback-Leibler (KL) divergence~\cite{Lin2018}:
\begin{equation}\label{eq:ML_optimisation_1}
\widehat{\mathbf{p}}_{\mathrm{ML}} = \argmin_{\mathbf{p}\in\mathcal{P}_{\mathrm{NS}}}\sum_{\mathbf{x}}\pi(\mathbf{x})D_{\mathrm{KL},\mathbf{x}}(\mathbf{f}\,\Vert\,\mathbf{p}),
\end{equation}
where $D_{\mathrm{KL},\mathbf{x}}\big(\mathbf f\,\Vert\,\mathbf p\big)\coloneqq\sum_{\mathbf{a}}f(\mathbf{a}\lvert\mathbf{x})\log_{2}\frac{f(\mathbf{a}\lvert\mathbf{x})}{p(\mathbf{a}\lvert\mathbf{x})}$ is the KL divergence of $\mathbf{f}\coloneqq\{f(\mathbf{a}\lvert\mathbf{x})\}$ with respect to $\mathbf{p}\coloneqq\{p(\mathbf{a}\lvert\mathbf{x})\}$ for a fixed $\mathbf{x}$. This minimisation is also referred to as a reverse $I$-projection (the $m$-projection in information geometry parlance) of $\mathbf f$ onto $\mathcal{P}_{\mathrm{NS}}$~\cite{Csiszar73,CsiszarMatus03,Nielsen2020}. In contrast to the $L^2$-method---which is a closed-form affine (orthogonal) projection onto $\mathcal{A}$---the ML estimator is obtained by solving a convex program: maximise a concave log-likelihood over the convex feasible set $\mathcal{P}_{\mathrm{NS}}$. Although convex, this remains an iterative constrained routine and can become taxing as the number of parties $n$ and the measurement-settings size $\abs{\mathcal{S}}=m$ grow. A further practical difference is that under ML, any instance $f(\mathbf{a}\lvert\mathbf{x})>0$ but $p(\mathbf{a}\lvert\mathbf{x})=0$ (whether forced by the linear equality/inequality constraints or arising at the boundary optimum) contributes $+\infty$ via the term $f\log_{2}(f/p)$, necessitating practical solvers to include smoothing or explicit positivity slacking; the $L^2$-approach, on the other hand, remains finite in such a situation requiring no such smoothing. Finally, while the ML method does literally correspond to the no-signalling distribution \textit{most likely} to have produced the given data under the assumption of a fixed underlying distribution sampled over successive identical experimental trials, in practice this assumption is never perfectly met. In particular, signalling effects often correspond to drifts in the underlying state that are amplified by (non-ideal, but common) experimental setups in which measurement settings are implemented in successive fixed blocks. Here, a distribution describing the experiment as a whole would be better conceived of as an average of the slightly different distributions occurring over the course of data acquisition; projections minimising Euclidean distance of probabilities, in the spirit of least-squares estimation, are a more natural choice in this paradigm. 

Having motivated our consideration of the $L^2$ projection approach, we now formulate it. Given $\mathbf{f}\in\mathbb{R}^d$, an $(n,m,2)$ behaviour that may exhibit small departures from no-signalling, the $L^2$-projection $\widehat{\mathbf{p}}=\Pi_{\mathcal{A}}(\mathbf{f})$ onto $\mathcal{A}=\mathrm{aff}(\mathcal{P}_{\mathrm{NS}})$ is given by
\begin{equation}\label{eq:AffineProj_definition}
\Pi_{\mathcal{A}}(\mathbf{f})\coloneqq\Pi_{\mathrm{ker}}(\mathbf{f}-\mathbf{d})+\mathbf{d},
\end{equation}
where $\Pi_{\mathrm{ker}}$ is the orthogonal projector onto $\mathrm{ker}(A_{\mathrm{eq}})\subset\mathbb{R}^d$ (the direction space parallel to $\mathcal A$), and the displacement vector $\mathbf{d}\in\mathcal{P}_{\mathrm{NS}}$ can be taken to be an (outcome-)uniform behaviour each component of which is defined as $d(\mathbf{a}\lvert\mathbf{x})=1/2^{n}$. The projected behaviour $\widehat{\mathbf{p}}$ satisfies the condition that among all the points in $\mathcal{A}$, it is the one closest in Euclidean distance to $\mathbf f$, i.e.,
\begin{equation}\label{eq:LeastSquares}
\widehat{\mathbf p}=\argmin_{\mathbf{p}\in\mathcal{A}}\norm{\mathbf{f}-\mathbf{p}}_{2}^{2}.
\end{equation}
The above least-squares formulation makes concrete the idea that we ``repair'' $\mathbf f$ by the smallest $L^2$-adjustment that restores the equalities defining the no-signalling conditions.

Expression~\eqref{eq:AffineProj_definition} yields a direct method of constructing the affine projector $\widehat{\mathbf{p}}=\Pi_{\mathcal{A}}(\mathbf{f})$: Find the orthogonal projector $\Pi_{\mathrm{ker}}$ by computing basis vectors of the linear subspace $\mathrm{ker}(A_{\mathrm{eq}})$, form the matrix $B$ with these vectors as columns, and obtain $\Pi_{\mathrm{ker}}$ as $B(B^{T}B)^{-1}B^{T}$; then, the affine projection is given by~\eqref{eq:AffineProj_definition}. Furthermore, the chosen $\mathbf{d}$ in~\eqref{eq:AffineProj_definition} is the minimum-norm point in $\mathcal A$ (the orthogonal representative), implying $\mathbf{d}\perp\mathrm{ker}(A_{\mathrm{eq}})$, i.e., $\Pi_{\mathrm{ker}}\mathbf{d}=\mathbf{0}$, and so~\eqref{eq:AffineProj_definition} simplifies to $\Pi_{\mathrm{ker}}\mathbf{f}+\mathbf{d}$. (Any displacement vector $\mathbf{v}\in\mathcal{A}$ would validate~\eqref{eq:AffineProj_definition}, but only the chosen $\mathbf{d}$ yields the stated simplification.)

The algorithm for computing $\Pi_{\mathcal{A}}(\mathbf{f})$, as just described, will generally be faster to implement than the iterative ML method described earlier. However, this method is potentially still somewhat computationally costly for several reasons: (1) the ambient dimension $d=(2m)^{n}$ grows rapidly for higher Bell scenarios (for even a modest increase in $n$ and $m$); (2) finding basis vectors of $\mathrm{ker}(A_{\mathrm{eq}})$ requires singular value decomposition or iterative null-space methods on a rank-deficient and highly sparse matrix ($A_{\mathrm{eq}}$), which typically results in basis vectors that are highly non-sparse in comparison to the rows of $A_{\mathrm{eq}}$; (3) the projector onto $\mathrm{ker}(A_{\mathrm{eq}})$ involves computing the inverse of the Gram matrix $(B^{T}B)^{-1}$. 

We also note that, in general, the $L^2$-projection need not yield a valid behaviour---the optimiser in~\eqref{eq:LeastSquares} can contain negative probabilities (see (B12) in Appendix B of~\cite{Lin2018} for an example). This pathology, however, only appears when statistical fluctuations are large enough to push $\mathbf f$ significantly outside $\mathcal{A}\cap\mathbb{R}_{+}^d$. In well-calibrated Bell tests operated in the high-count regime, where the empirical frequencies already lie very close to the true underlying no-signalling behaviour, violations like the example of \cite{Lin2018}, with signalling probability differences as large as 40\%, are not observed in practice. If there are no negative probabilities in the estimate obtained from \eqref{eq:LeastSquares}, we can proceed with it; otherwise we enforce non-negativity by replacing the constraint in~\eqref{eq:LeastSquares} with $\mathbf{p}\in\mathcal{A}\cap\mathbb{R}_{+}^d$ which modifies the optimisation to a convex quadratic program. It is still computationally simpler and numerically stabler than the ML routine as the scenario size grows (more number of parties and inputs). In fact, one can arrive at a simple sufficient condition under which non-negativity is nevertheless guaranteed. Let $\mathbf{p}_0\in\mathcal{P}_{\mathrm{NS}}$ be the underlying no-signalling behaviour and suppose $\mathbf{p}_0$ is strictly positive with $\gamma\coloneqq \min_{\mathbf{a},\mathbf{x}}p_0(\mathbf{a}|\mathbf{x})>0$. Writing the empirical behaviour as $\mathbf{f}=\mathbf{p}_0+\boldsymbol{\epsilon}$, linearity of the affine projection $\Pi_{\mathcal A}$ around $\mathbf{p}_0$ gives
\begin{equation*}
\Vert\Pi_{\mathcal A}(\mathbf{f})-\mathbf{p}_0\Vert_2 = \Vert\Pi_{\mathrm{ker}}(\mathbf{f}-\mathbf{p}_0)\Vert_2\le\Vert\mathbf{f}-\mathbf{p}_0\Vert_2=\Vert\boldsymbol{\epsilon}\Vert_2.
\end{equation*}

The first equality holds because $\mathcal{A}=\mathbf{d}+\mathrm{ker}(A_\mathrm{eq})$, and since $\mathbf{p}_0,\mathbf{d}\in\mathcal{A}$, the difference $\mathbf{p}_0-\mathbf{d}$ belongs to $\mathrm{ker}(A_{\mathrm{eq}})$, i.e., $\Pi_{\mathrm{ker}}(\mathbf{p}_0-\mathbf{d})=\mathbf{p}_0-\mathbf{d}$. Using this with a little manipulation one can show that $\Pi_{\mathcal A}(\mathbf{f})-\mathbf{p}_0 = \Pi_{\mathrm{ker}}(\mathbf{f}-\mathbf{d})-(\mathbf{p}_0-\mathbf{d})=\Pi_{\mathrm{ker}}(\mathbf{f}-\mathbf{d}-\mathbf{p}_0+\mathbf{d})$. The inequality then follows from the fact that $\Pi_{\mathrm{ker}}$ is an orthogonal projector, hence it is norm contracting, i.e., $\Vert\Pi_{\mathrm{ker}}(\mathbf{f}-\mathbf{p}_0)\Vert_2\le\Vert\mathbf{f}-\mathbf{p}_0\Vert_2$. Hence, a sufficient condition for positivity is
\begin{equation*}
\Vert\mathbf{f}-\mathbf{p}_0\Vert_2 < \gamma,\,\text{where }\gamma\coloneqq\min_{\mathbf{a},\mathbf{x}}p_{0}(\mathbf{a}|\mathbf{x})>0,
\end{equation*}
using the fact that the $\Vert\cdot\Vert_2$-distance between $\mathbf f$ and $\mathbf{p}_0$ is an upper bound on the difference for any single component of the two vectors. This is a conservative but useful condition. It means if $\mathbf{f}=\mathbf{p}_0+\boldsymbol{\epsilon}$, for some perturbation $\boldsymbol{\epsilon}$, then no negative entry can occur whenever every true probability (i.e., every component $p_0(\mathbf{a}|\mathbf{x})$) of $\mathbf{p}_0$ is separated from zero by more than the empirical $L^2$-error (which is $\Vert\mathbf{f}-\mathbf{p}_0\Vert_2=\Vert\boldsymbol{\epsilon}\Vert_2$), in other words, if the perturbation is small enough. In Section~\ref{ss:FiniteSampleProbBound}, we provide a finite-sample high-probability bound in the spirit of Hoeffding/McDiarmid~\cite{Hoeffding_63,McDiarmid_1989} to the possibility that a projected empirical probability vector has negative coordinates---the relevant details to the results presented therein are included in Appendix~\ref{a:NegativeProbIssue}.

Even after enforcing non-negativity, it is possible the projected behaviour might lie outside the quantum-achievable (convex) set $\mathcal{P}_{\mathrm Q}$, in which case the optimisation in~\eqref{eq:LeastSquares} should be solved with the constraint $\mathbf{p}\in\mathcal{P}_{\mathrm Q}$. Appendix C of~\cite{Lin2018} shows that the unique minimiser of this program is equivalent to first applying the projection onto $\mathcal{A}$ followed by minimising the $L^2$-distance from that projection to a target feasible set (see Lemma 1 therein). More precisely, if one replaces the (intractable) quantum set by any converging NPA-type outer relaxation $\mathcal{P}_{\mathrm{Q}_k}$~\cite{Navascués_2008}, such that $\mathcal{P}_{\mathrm Q}\subset\mathcal{P}_{\mathrm{Q}_k}\subset\mathcal{P}_{\mathrm{NS}}\subset\mathcal{A}$, the least-squares problem reduces to a semi-definite program~\cite{SkrzypczykCavalcanti_23}. A detailed numerical comparison between the projection-based, least-square, and maximum-likelihood point estimators for finite Bell data has been carried out in~\cite{Lin2018}. In particular, their Fig.~3 compares the affine projection method and maximum-likelihood estimators (along with some other variant estimators) for several representative quantum behaviours from the $(2,2,2)$ Bell scenario over $10^4$ simulation runs and sample sizes ranging from $N_{\text{trials}}=10^2$ to $10^{10}$. The reported mean $1$-norm deviation from the generating distribution decreases approximately as $N_{\text{trials}}^{-1/2}$ for all these estimators, including the projection method and the ML-type estimators. Their results also illustrate the expected trade-off: ML estimators produce physical estimates and can be preferable for statistical point estimation, whereas the affine projection is closed-form and computationally inexpensive but may fail to preserve non-negativity in low-statistics or boundary regimes.

The next section presents a closer examination of the equality-constraint matrix $A_{\mathrm{eq}}$. We take advantage of the sparsity of $A_{\mathrm{eq}}$ in deriving results in subsequent sections culminating in a computationally simpler projection realised as the composition of three pre-computable linear maps.

\subsection{Row structure and sparsity of \texorpdfstring{$A_{\mathrm{eq}}$}{Aeq}}
Recalling that $A_{\mathrm{eq}}$ has $d=(2m)^n$ columns, any row vector $\mathbf{r}\in\mathbb{R}^{d}$ of $A_{\mathrm{eq}}$ acts as $\mathbf{r}^T\mathbf{p}=\sum_{(\mathbf{a},\mathbf{x})}r(\mathbf{a},\mathbf{x})p(\mathbf{a}\lvert\mathbf{x})$. Revisiting then the no-signalling conditions in~\eqref{eq:NoSignm2}, for each party $i\in[n]$, each non-reference input $x_{i}\in\mathcal{S}\setminus\{\mathsf{s}_{1}\}$, and each fixed outcome-settings combination $(\tilde{\mathbf{a}}_{\neg i},\tilde{\mathbf{x}}_{\neg i})\in\mathcal{O}
^{n-1}\times\mathcal{S}^{n-1}$, the equality
\begin{equation*}
p(\tilde{\mathbf{a}}_{\neg i}\lvert\tilde{\mathbf{x}}_{\neg i},x_{i}=\mathsf{s}_{1})-p(\tilde{\mathbf{a}}_{\neg i}\lvert\tilde{\mathbf{x}}_{\neg i},x_{i}=\mathsf{s}_{r})=0
\end{equation*}
can be written as $\big(\mathbf{ns}_{(\tilde{\mathbf{a}}_{\neg i},\tilde{\mathbf{x}}_{\neg i})}^{i,r}\big)^T\mathbf{p} = 0$, where the corresponding row is defined componentwise as shown below:
\begin{align}\label{eq:NoSigVec_nm2}
&\text{For any fixed }(\tilde{\mathbf{a}}_{\neg i},\tilde{\mathbf{x}}_{\neg i})\in\mathcal{O}^{n-1}\times\mathcal{S}^{n-1},\nonumber \\
&\mathrm{ns}_{(\tilde{\mathbf{a}}_{\neg i},\tilde{\mathbf{x}}_{\neg i})}^{i,r}(\mathbf{a},\mathbf{x}) \coloneqq \delta_{\mathbf{a}_{\neg i},\tilde{\mathbf{a}}_{\neg i}}\delta_{\mathbf{x}_{\neg i},\tilde{\mathbf{x}}_{\neg i}}\left(\delta_{x_{i},\mathsf{s}_{1}}-\delta_{x_{i},\mathsf{s}_{r}}\right), \forall i\in[n],r\in[m]\setminus\{1\},(\mathbf{a},\mathbf{x}).
\end{align}
Because the outcomes are binary, each such row has exactly four non-zero entries (the coefficients): for the two entries with $(\mathbf{x}_{\neg i},x_i)=(\tilde{\mathbf{x}}_{\neg i},\mathsf{s}_{1})$ and $a_{i}\in\{0,1\}$, the coefficient is $+1$; for the two entries $(\mathbf{x}_{\neg i},x_i)=(\tilde{\mathbf{x}}_{\neg i},\mathsf{s}_{r})$ and $a_{i}\in\{0,1\}$, the coefficient is $-1$. The triple $(i,r,(\tilde{\mathbf{a}}_{\neg i},\tilde{\mathbf{x}}_{\neg i}))$ therefore uniquely labels each no-signalling row. 

Next, for any fixed realisation $\tilde{\mathbf{x}}$ of the settings combination $\mathbf{x}$, the normalisation condition $\sum_{\mathbf{a}}p(\mathbf{a}\lvert\tilde{\mathbf{x}})=1$ can be represented by $\big(\mathbf{nrm}_{\tilde{\mathbf{x}}}\big)^T\mathbf{p}=1$, where the row vector $\mathbf{nrm}_{\tilde{\mathbf{x}}}\in\mathbb{R}^{d}$ is defined componentwise as shown below.
\begin{equation}\label{eq:Normalisation_vector}
\text{For any fixed } \tilde{\mathbf{x}}\in\mathcal{S}^{n}, \mathrm{nrm}_{\tilde{\mathbf{x}}}(\mathbf{a},\mathbf{x})\coloneqq \delta_{\mathbf{x},\tilde{\mathbf{x}}}, \forall(\mathbf{a},\mathbf{x}).
\end{equation}
Thus, $\mathbf{nrm}_{\tilde{\mathbf{x}}}$ has exactly $2^n$ non-zero entries, all equal to $1$, supported on the coordinates $\{(\mathbf{a},\tilde{\mathbf{x}})\colon\mathbf{a}\in\mathcal{O}^n\}$. These $m^n$ normalisation rows are mutually linearly independent.

With a fixed ordering of the rows where we put the $n(m-1)(2m)^{n-1}$ no-signalling rows first and the $m^n$ normalisation rows last, we may write
\begin{equation*}
A_{\mathrm{eq}} = \begin{pmatrix}
A_{\mathrm{ns}} \\ A_{\mathrm{no}}
\end{pmatrix},\, \mathbf{b} = \begin{pmatrix}
\mathbf{0} \\ \mathbf{1}
\end{pmatrix},
\end{equation*}
where $A_{\mathrm{ns}}$ stacks the vectors $\mathbf{ns}_{(\tilde{\mathbf{a}}_{\neg i},\tilde{\mathbf{x}}_{\neg i})}^{i,r}$ and $A_{\mathrm{no}}$ stacks the vectors $\mathbf{nrm}_{\tilde{\mathbf{x}}}$, and $\mathbf{0}$ and $\mathbf{1}$ represent vectors comprising $0$'s and $1$'s, respectively. The sparsity pattern is particularly simple: every no-signalling row has exactly four non-zeros $(+1,+1,-1,-1)$, while every normalisation row has $2^n$ non-zero values that are all $+1$. This sparsity pattern will be useful in showing that the uniformly-averaged marginal correlator coefficient vectors---which we describe in the next section---are orthogonal to every row of $A_{\mathrm{eq}}$, and hence lie in its kernel.

\subsection{Uniformly-averaged marginal correlators}\label{s:CorrsFornm2BellScenario}
For a strictly no-signalling behaviour, i.e., $\mathbf{p}\in\mathcal{P}_{\mathrm{NS}}$, the $k$-party marginal correlator is unambiguous: it is the expectation of the product of the ($\pm 1$-valued) outcomes of the $k$ parties, independent of the settings choices of the remaining $n-k$ parties. For empirical behaviours that may exhibit weak deviation from no-signalling, however, this quantity is not uniquely determined---different ways of conditioning or averaging over the complementary parties' settings choices lead to inequivalent functionals. In this work, we use a canonical, symmetry-respecting choice of functional in $\mathbf{p}$ that (1) is linear in $\mathbf{p}$, (2) is invariant under permutation of the settings choices of the parties outside the $k$-party set, (3) reduces to the standard marginal correlator whenever no-signalling holds, and, notably, (4) is invariant under the $L^2$-projection onto $\mathcal{A}$, as shown later in Proposition~\ref{prop:CorrVecBelongsToNullSpace}. Namely, we take all of the possible different values that the $k$-party-subset correlator can assume -- for each fixed setting choice of the remote, complementary $n-k$ parties -- and uniformly average them; examples of these are used implicitly in Appendix D of~\cite{Bancal_2014} with more explicit discussion in Appendix B.3 of~\cite{Lin2018}. We refer to this as the $k$-party \textit{uniformly-averaged marginal correlator} (UMC) and formalise it below in~\eqref{eq:Marginal_Corr}. The Bell scenario considered here involves \textit{binary} outcomes for all settings and all parties. In general, the machinery can be extended to an arbitrary number of outcomes, though at the cost of more complicated and less natural-looking formulas for correlators. This more general approach is adopted in \cite{Rosset2020}.

Before giving the definition, we first specify the notational conventions underlying a $k$-party UMC. For $k\in[n]$, let $I\coloneqq\{i_1,i_2,\ldots,i_k\}\subseteq[n]$, with $1\le i_1<i_2<\cdots<i_k\le n$, be a $k$-element index set. Denote by $\mathbf{x}_{I}$ the settings choice $(x_i)_{i\in I}$ of the $k$ parties in $I$ and by $\mathbf{x}_{\bar I}$ the settings choice $(x_{j})_{j\in \bar{I}}$ for the remaining $n-k$ parties so that $\mathbf{x}\equiv(\mathbf{x}_{I},\mathbf{x}_{\bar I})$, for $\bar{I}=[n]\setminus I$; likewise for $\mathbf{a}$. The $k$-party UMC is then the linear functional $\bar{C}_{\mathbf{x}_{I}}^{I}\colon\mathbb{R}^{d}\to\mathbb{R}$ defined for all $\mathbf{x}_{I}\in\mathcal{S}^{k}$ as
\begin{equation}
\bar{C}_{\mathbf{x}_{I}}^{I}(\mathbf{p})\coloneqq\frac{1}{m^{n-|I|}}\sum_{\mathbf{x}_{\bar I}}\sum_{\mathbf{a}}\chi_{I}(\mathbf{a})p(\mathbf{a}\lvert\mathbf{x}),\label{eq:Marginal_Corr}
\end{equation}
where $\chi_{I}(\mathbf{a})\coloneqq(-1)^{\oplus_{i\in I}a_i}\in\{+1,-1\}$ is the parity function. The $k$-party UMC in~\eqref{eq:Marginal_Corr} is a linear combination of the components of $\mathbf{p}$. It can therefore be expressed as $\big(\mathbf{c}_{\tilde{\mathbf{u}}_{I}}^I\big)^T\mathbf{p}$, where the correlator coefficient vector $\mathbf{c}_{\tilde{\mathbf{u}}_{I}}^I\in\mathbb{R}^d$, for $d=(2m)^n$, is defined for any fixed settings combination $\tilde{\mathbf{u}}_{I}\in\mathcal{S}^{k}$ of the $k$-party set. The components of $\mathbf{c}_{\tilde{u}_I}^I$ are:
\begin{align}\label{eq:MarginalCorrvec}
c^{I}_{\tilde{\mathbf{u}}_{I}}(\mathbf{a},\mathbf{x}) \coloneqq \frac{1}{m^{n-|I|}}\chi_{I}(\mathbf{a})\delta_{\mathbf{x}_{I},\tilde{\mathbf{u}}_{I}}.
\end{align}

For $I=[n]$, the expression in~\eqref{eq:Marginal_Corr} corresponds to a (standard) full $n$-party correlator:
\begin{equation}\label{eq:Full_Corr}
\bar{C}_{\mathbf{x}}(\mathbf{p}) = \sum_{\mathbf{a}}(-1)^{\bigoplus_{i=1}^{n}a_{i}}p(\mathbf{a}\lvert\mathbf{x}),\,\forall \mathbf{x}\in\mathcal{S}^{n}.
\end{equation}
The expression $\mathbf{c}_{\tilde{\mathbf{u}}}^T\mathbf{p}$ equivalently represents the full correlator in~\eqref{eq:Full_Corr}, where $\mathbf{c}_{\tilde{\mathbf{u}}}$ is defined analogously to~\eqref{eq:MarginalCorrvec} for any fixed settings combination $\tilde{\mathbf{u}}\in\mathcal{S}^{n}$ of the $n$ parties.

Before proceeding, we make note of a representation choice. Following~\cite{Rosset2020}, one can impose an order on $\mathcal{O}^n\times\mathcal{S}^n$ and work in the corresponding Kronecker basis, which is convenient for algebra on local maps (and accommodates our $k$-party UMC). Here we adopt an ordering-free description, where constraints are given by supports and signs only. We remark that while the results of the remainder of this section can emerge from the comprehensive treatment in~\cite{Rosset2020}, they are not presented there in an explicit manner and our direct approach avoids its more technical machinery.

With the uniformly-averaged correlator coefficient vectors and the row structure of $A_{\mathrm{eq}}$ in place, we now state and prove a key proposition that enables a simplified computational method for obtaining the 
$L^2$-projection.

\begin{prop}\label{prop:CorrVecBelongsToNullSpace}
For all $\tilde{\mathbf{u}}_{I}\in\mathcal{S}^{k}$, the UMC coefficient vector $\mathbf{c}^{I}_{\tilde{\mathbf{u}}_{I}}$ satisfies $\mathbf{c}^{I}_{\tilde{\mathbf{u}}_{I}}\in\mathrm{ker}(A_{\mathrm{eq}})$.
\end{prop}
\begin{proof}
Proof in Appendix~\ref{a:ProofProp1}
\end{proof}
\begin{cor}[Invariance of \texorpdfstring{$k$}{k}-party UMC under \texorpdfstring{$\Pi_\mathcal{A}$}{Pi_A}]\label{cor:InvarianceUnderProjection_nm2}
Given any behaviour $\mathbf p\in\mathbb{R}^{d}$ not necessarily adhering to no-signalling, the $L^2$-projection onto the no-signalling affine hull $\mathcal A$ preserves the value of all $k$-party UMC coefficient vectors, i.e., $\big(\mathbf{c}^{I}_{\tilde{\mathbf{u}}_{I}}\big)^T\mathbf{p}=\big(\mathbf{c}^{I}_{\tilde{\mathbf{u}}_{I}}\big)^T\Pi_{\mathcal A}(\mathbf{p})$.
\end{cor}
\begin{proof}
The proof follows straightforwardly from the definition of $\Pi_{\mathcal{A}}$ given in~\eqref{eq:AffineProj_definition}. We fix $I$ and $\tilde{\mathbf{u}}_{I}$, and abbreviate $\mathbf{c}\coloneqq\mathbf{c}_{\tilde{\mathbf{u}}_{I}}^{I}$, since all statements below hold uniformly for any choice of $(I, \tilde{\mathbf{u}}_{I})$.
\begin{align}
\mathbf{c}^T\Pi_{\mathcal A}(\mathbf{p})
=&\mathbf{c}^T\Pi_{\mathrm{ker}}\mathbf{p} - \mathbf{c}^T\Pi_{\mathrm{ker}}\mathbf{d} + \mathbf{c}^T\mathbf{d}
=\big(\Pi_{\mathrm{ker}}\mathbf{c}\big)^T\mathbf{p}-\big(\Pi_{\mathrm{ker}}\mathbf{c}\big)^T\mathbf{d} + \mathbf{c}^T\mathbf{d}\nonumber \\
=&\mathbf{c}^T\mathbf{p} - \mathbf{c}^T\mathbf{d} + \mathbf{c}^T\mathbf{d}
=\mathbf{c}^T\mathbf{p}.\label{eq:InvUnderProj_proof}
\end{align}
The second and third equalities in~\eqref{eq:InvUnderProj_proof} follow, respectively, from the property that as a projection $\Pi_{\mathrm{ker}}^T=\Pi_{\mathrm{ker}}$ (clear from the formula $\Pi_{\mathrm{ker}}=B(B^TB)^{-1}B^T$), and the fact that $\Pi_{\mathrm{ker}}$ acts as an identity on $\mathbf{c}$ (since $\Pi_{\mathrm{ker}}$ is the unique self-adjoint idempotent operator with $\mathrm{range}(\Pi_{\mathrm{ker}})=\mathrm{ker}(A_{\mathrm{eq}})$). Indeed, $\mathbf{c}$ is expressible as $\mathbf{c}=B\mathbf{x}$, a linear combination of the columns of $B$ (which are basis vectors of $\mathrm{ker}(A_{\mathrm{eq}})$). Then using the explicit orthogonal projector $\Pi_{\mathrm{ker}}=B(B^{T}B)^{-1}B^T$ on $\mathbf{c}$ we get $\Pi_{\mathrm{ker}}\mathbf{c}=B(B^{T}B)^{-1}B^{T}B\mathbf{x}=B\mathbf{x}=\mathbf{c}$.
\end{proof}
Although not necessary for constructing an algorithm for the $L^2$ projection, the following results strengthen Proposition~\ref{prop:CorrVecBelongsToNullSpace} by establishing that the UMC coefficient vectors in fact form a \textit{basis} for $\mathrm{ker}(A_{\mathrm{eq}})$.
\begin{cor}[UMC coefficient vectors span $\mathrm{ker}(A_{\mathrm{eq}})$]\label{cor:UMCvecsSpanKernel}
Proposition~\ref{prop:CorrVecBelongsToNullSpace} shows that every uniformly-averaged marginal correlator coefficient vector belongs to $\mathrm{ker}(A_{\mathrm{eq}})$. In fact, these vectors span the entire subspace $\mathrm{ker}(A_{\mathrm{eq}})$, i.e.
\begin{equation}\label{eq:KernelA_equals_SpanUMC}
\mathrm{ker}(A_{\mathrm{eq}}) = \mathrm{span}\{\mathbf{c}_{\mathbf{u}_I}^I\colon\varnothing\neq I\subseteq[n],\mathbf{u}_I\in\mathcal{S}^{|I|}\}.
\end{equation}
Equivalently, every element of $\mathrm{ker}(A_{\mathrm{eq}})$ is expressible as a linear combination of the UMC coefficient vectors.
\end{cor}
\begin{proof}
Since Proposition~\ref{prop:CorrVecBelongsToNullSpace} implies $\mathrm{span}\{\mathbf{c}_{\mathbf{u}_I}^I\colon\varnothing\neq I\subseteq[n],\mathbf{u}_I\in\mathcal{S}^{|I|}\}\subseteq\mathrm{ker}(A_{\mathrm{eq}})$, we just need to show the reverse inclusion. Let $\mathbf{v}\in\mathrm{ker}(A_{\mathrm{eq}})$. We view $\mathbf{v}$ as an array with components $v(\mathbf{a}|\mathbf{x})$, where $\mathbf{a}\in\{0,1\}^n$ and $\mathbf{x}\in\mathcal{S}^n$. For each fixed setting $\mathbf{x}\in\mathcal{S}^n$, the map $\mathbf{a}\mapsto v(\mathbf{a}|\mathbf{x})$ is a vector with components indexed by $\mathbf a$. It can also be viewed as a real-valued function on $\{0,1\}^n$, which can be expanded in the basis:
\begin{equation*}
\{\chi_I\colon\{0,1\}^n\to\mathbb{R}\}_{I\subseteq[n]},\,\text{where }\chi_I(\mathbf a)\coloneq(-1)^{\bigoplus_{i\in I}a_i}.
\end{equation*}
The above set (of parity functions) forms a basis because it satisfies orthogonality with respect to the inner product $\langle f,g\rangle=\sum_{\mathbf a}f(\mathbf a)g(\mathbf a)$. Although we only consider non-empty subsets $I$ of $[n]$ for the basis elements $\chi_I$, for $I=\varnothing$, the function $\chi_I$ maps everything to $1$. Indeed,
\begin{equation*}
\sum_{\mathbf a}\chi_I(\mathbf a)\chi_J(\mathbf a) = \sum_{\mathbf a}\chi_{I\Delta J}(\mathbf a) =  2^n\delta_{I,J},
\end{equation*}
where $I\Delta J$ is the symmetric difference. (If $I\neq J$, then $I\Delta J\neq\varnothing$, and $\chi_{I\Delta J}$ takes $+1$ and $-1$ equally often over $\{0,1\}^n$, so the sum is zero. Since there are $2^n$ subsets $I\subseteq[n]$, and $\mathrm{dim}(\mathbb{R}^{\{0,1\}^n}) = |\{0,1\}^n| = 2^n$, orthogonality implies that these $2^n$ functions form a basis.) Hence, for each fixed setting $\mathbf{x}\in\mathcal{S}^n$, $v(\cdot|\mathbf{x})$ admits the expansion
\begin{equation}\label{eq:Expansion1}
v(\mathbf{a}|\mathbf{x}) = \frac{1}{2^n}\sum_{I\subseteq[n]}\chi_I(\mathbf{a})\widehat{v}_I(\mathbf{x}),
\end{equation}
where $\widehat{v}_I(\mathbf x)\coloneqq\sum_{\mathbf a}\chi_I(\mathbf a)v(\mathbf{a}|\mathbf{x})$. Since $\mathbf{v}\in\mathrm{ker}(A_{\mathrm{eq}})$, it satisfies $A_{\mathrm{eq}}\mathbf{v}=\mathbf{0}$, which---due to the inclusion of the normalisation conditions in $A_{\mathrm{eq}}$---means it satisfies $\sum_{\mathbf a}v(\mathbf{a}|\mathbf{x})=0$, for each $\mathbf x$, from which we conclude that $\widehat{v}_\varnothing(\mathbf x)=0$. Another condition that $\mathbf v$ satisfies is~\eqref{eq:NoSignm2}, which we now use to show that $\widehat{v}_I(\mathbf x)$ only depends on $\mathbf{x}_I$ (and not on the complementary setting combination $\mathbf{x}_{\bar{I}}$) whenever we view $\mathbf{x}$ as $(\mathbf{x}_I,\mathbf{x}_{\bar{I}})$. Fix a non-empty subset $I\subseteq[n]$. Suppose $j\notin I$. For two choices $\mathsf{s}_1,\mathsf{s}_r$
of the $j$'th setting, keeping all other settings fixed, we have
\begin{equation*}
\widehat{v}_I(\mathbf{x}_{\neg j},x_j=\mathsf{s}_1) - \widehat{v}_I(\mathbf{x}_{\neg j},x_j=\mathsf{s}_r) = \sum_{\mathbf a}\chi_I(\mathbf a)[v(\mathbf{a}|\mathbf{x}_{\neg j},x_j=\mathsf{s}_1)-v(\mathbf{a}|\mathbf{x}_{\neg j},x_j=\mathsf{s}_r)].
\end{equation*}
Since $j\notin I$, $\chi_I(\mathbf a)$ is independent of $a_j$. Therefore,
\begin{align*}
\widehat{v}_I(\mathbf{x}_{\neg j},x_j=\mathsf{s}_1) - \widehat{v}_I(\mathbf{x}_{\neg j},x_j=\mathsf{s}_r) = \sum_{\mathbf{a}_{\neg j}}\chi_{I}(\mathbf{a}_{\neg j})\underbrace{\sum_{a_j}[v(\mathbf{a}|\mathbf{x}_{\neg j},x_j=\mathsf{s}_1)-v(\mathbf{a}|\mathbf{x}_{\neg j},x_j=\mathsf{s}_r)]}_{=\,0\text{ (since }\mathbf{v}\text{ satisfies }\eqref{eq:NoSignm2})}.
\end{align*}
Hence, $\widehat{v}_I(\mathbf{x}_{\neg j},x_j=\mathsf{s}_1) = \widehat{v}_I(\mathbf{x}_{\neg j},x_j=\mathsf{s}_r)$. Since this holds for every $j\notin I$, $\widehat{v}_I(\mathbf x)$ is independent of all settings outside $I$. Thus there exists a function $\theta_I\colon\mathcal{S}^{|I|}\to\mathbb{R}$ such that $\widehat{v}_I(\mathbf x)=\theta_I(\mathbf{x}_I)$. Substituting this and the fact that $\widehat{v}_\varnothing(\mathbf x)=0$ into the expression of $v(\mathbf{a}|\mathbf{x})$ in~\eqref{eq:Expansion1}, we get for each $\mathbf x$:
\begin{equation}\label{eq:Expansion2}
v(\mathbf{a}|\mathbf{x})=\frac{1}{2^n}\sum_{\varnothing\neq I\subseteq[n]}\chi_I(\mathbf{a})\theta_{I}(\mathbf{x}_I) = \frac{1}{2^n}\sum_{\varnothing\neq I\subseteq[n]}\sum_{\mathbf{u}_I\in\mathcal{S}^{|I|}}\chi_I(\mathbf{a})\theta_I(\mathbf{u}_I)\delta_{\mathbf{x}_I,\mathbf{u}_I}.
\end{equation}
Using the definition of the UMC coefficient vector in~\eqref{eq:MarginalCorrvec}, we can rewrite~\eqref{eq:Expansion2} as:
\begin{equation}\label{eq:Expansion3}
\mathbf{v} = \sum_{\varnothing\neq I\subseteq[n]}\sum_{\mathbf{u}_I\in\mathcal{S}^{|I|}}\lambda_{I,\mathbf{u}_I}\mathbf{c}_{\mathbf{u}_I}^I,\,\text{where }\lambda_{I,\mathbf{u}_I} = \frac{m^{n-|I|}}{2^n}\theta_I(\mathbf{x}_I).
\end{equation}
This shows that every $\mathbf{v}\in\mathrm{ker}(A_{\mathrm{eq}})$ lies in the span of the UMC coefficient vectors, proving the required inclusion. This establishes the desired equality in~\eqref{eq:KernelA_equals_SpanUMC}.
\end{proof}

The UMC coefficient vectors provide an explicit correlator-based spanning family for all perturbations that preserve the homogeneous normalisation conditions and the no-signalling conditions. It is useful to compare this with the known dimension of the no-signalling polytope. Since $\mathcal{A}=\mathrm{aff}(\mathcal{P}_{\mathrm{NS}})$, and $\mathrm{ker}(A_{\mathrm{eq}})$ is the direction space of $\mathcal A$, we have $\mathrm{dim}(\mathrm{ker}(A_{\mathrm{eq}}))=\mathrm{dim}(\mathcal A)=\mathrm{dim}(\mathcal{P}_{\mathrm{NS}})$. Specialising the standard dimension formula for the no-signalling polytope of any Bell scenario (see Thm 3.1 in~\cite{Pironio2004PhD}) to the $(n,m,2)$ Bell scenario, we get $\mathrm{dim}(\mathcal{P}_{\mathrm{NS}})=(m+1)^n-1$. On the other hand, the number of non-empty coefficient vectors is obtained by choosing a non-empty subset of the $n$ parties, and then choosing one of $m$ settings for each party in the subset. The total number is
\begin{equation}\label{eq:Cardinality_UMCvecs}
|\{\mathbf{c}_{\mathbf{u}_I}^I\colon\varnothing\neq I\subseteq[n],\mathbf{u}_I\in\mathcal{S}^{|I|}\}| = \sum_{k=1}^{n}\binom{n}{k}m^k=(m+1)^n-1
\end{equation}
So the UMC coefficient vectors are not merely a spanning family; they have exactly the same number of elements as $\mathrm{dim}(\mathrm{ker}(A_{\mathrm{eq}}))$. Indeed, they form a basis of $\mathrm{ker}(A_{\mathrm{eq}})$.

\section{\label{s:ClosedFormProj}Closed formula for projection}
We now use the results of the previous section to present a simplified, closed-form approach to constructing the projector onto $\mathcal A$. An important step in the formulation of this projector is a linear (inverse) relation that expresses a no-signalling behaviour in terms of the correlators derived from it. We first derive this inverse relation.

For any subset of parties indexed by $I \subseteq [n]$, we define the correlator $C_{\mathbf{x}_{I}}^{I}$ for a no-signalling behaviour $\mathbf{p}$ as shown below: 
\begin{equation}\label{eq:nm2Corr}
C_{\mathbf{x}_{I}}^{I}\coloneqq \sum_{\mathbf{a}_{I}}\chi_I(\mathbf{a})p(\mathbf{a}_{I}\lvert\mathbf{x}_{I}),
\end{equation}
where $\chi_I(\mathbf{a})=(-1)^{\bigoplus_{l\in I}a_l}$; recall that the vector $\mathbf{x}_{I}$ denotes the vector of the inputs for the $k$-parties, i.e., $\mathbf{x}_{I}\equiv (x_{i_1},x_{i_2},\ldots,x_{i_k})$ such that $\mathbf{x}\equiv(\mathbf{x}_I,\mathbf{x}_{\bar I})$. Due to the no-signalling condition, $p(\mathbf{a}_{I}\lvert\mathbf{x}_{I}) = p(\mathbf{a}_{I}\lvert\mathbf{x}_I,\tilde{\mathbf u}_{\bar I})$
for any arbitrary fixed choice $\tilde{\mathbf u}_{\bar I}$ of the non-$k$-party input $\mathbf{x}_{\bar I}$. And so the expression in~\eqref{eq:nm2Corr} can be equivalently written as follows:
\begin{equation}\label{e:corr}
C_{\mathbf{x}_{I}}^{I} = \sum_{\mathbf{a}}\chi_I(\mathbf{a})p(\mathbf{a}\lvert\mathbf{x}),
\end{equation}
where $\mathbf x$ can be input combination that agrees with the $k$-party input $\mathbf{x}_I$. The summation in~\eqref{e:corr} is over all outcome combination $\mathbf{a}$, in which $\mathbf{a}_I$ is involved in determining the sign in front of $p(\mathbf{a}\lvert\mathbf{x})$ through the parity function, and the rest of the $\mathbf{a}_{\bar I}$ are marginalised over. Following standard expressions for the inverse for probabilities in terms of correlators (as given in expression (11) in~\cite{Goh_2018} and expression (4) in~\cite{Pironio_2011} for the $(2,2,2)$ and $(3,2,2)$ Bell scenarios, respectively), the generalisation to a scenario involving an arbitrary $n$ number of parties, each performing an arbitrary $m$ number of dichotomic measurements, yields the following formula:
\begin{equation}\label{e:prob}
p(\mathbf{a}\lvert\mathbf{x}) = \frac{1}{2^{n}}\sum_{I\subseteq[n]}\chi_I(\mathbf{a})C_{\mathbf{x}_{I}}^{I}.
\end{equation}
Notice that the empty set $I = \varnothing$ is included in the sum, which is performed over \textit{all} subsets of any cardinality. For the correlator corresponding to $I=\varnothing$, we have $\bar{C}^{\varnothing}=1$, since the quantity $\bigoplus_{l\in I}a_l$ is equal to zero for the empty set and so in~\eqref{e:corr} $\bar{C}^{\varnothing}$ is the sum (over all combinations of $\mathbf{a}$) of the conditional probabilities, which equals $1$.

The validity of~\eqref{e:prob} can be seen by plugging~\eqref{e:corr} into the right side of~\eqref{e:prob} after first fixing a choice of outcomes $\mathbf{a}'\equiv(a'_{1},...,a'_{n})$, yielding
\begin{equation}\label{e:proof2}
p(\mathbf{a}'\lvert\mathbf{x})=\frac{1}{2^n}\sum_{I\subseteq [n]} \chi_I(\mathbf{a}')\sum_{\mathbf{a}}\chi_I(\mathbf{a})p(\mathbf{a}\lvert\mathbf{x})=\frac{1}{2^n}\sum_{\mathbf{a}}p(\mathbf{a}\lvert\mathbf{x})\underbrace{\sum_{I\subseteq[n]}\chi_I(\mathbf{a}')\chi_I(\mathbf{a})}_{\coloneqq S(\mathbf{a}',\mathbf{a})}.
\end{equation}
Now there are $2^n$ subsets of $[n]$. We then observe that the inner sum in~\eqref{e:proof2} satisfies
\begin{align*}
S(\mathbf{a}',\mathbf{a}) = \sum_{I\subseteq[n]}(-1)^{\bigoplus_{l\in I}(a'_{l}\oplus a_l)} = \begin{cases}
    2^n, & \mathbf{a}'=\mathbf{a},\\
    0, & \mathbf{a}'\neq\mathbf{a}.
\end{cases}
\end{align*}
The above identity for $S(\mathbf{a}',\mathbf{a})$ can be shown as follows: First, we have
\begin{equation*}
    (-1)^{\bigoplus_{l\in I}(a_{l}'\oplus a_l)}=\prod_{l\in I}(-1)^{a_{l}'\oplus a_l}.
\end{equation*}
Plugging this in the expression for $S(\mathbf{a}',\mathbf{a})$ and then using the fact that expanding $\prod_{l=1}^n(1+a_l)$ for a set of numbers $\{a_1,...,a_n\}$ yields precisely the sum of all terms of the form $\prod_{l\in I}a_l$ for each subset $I \subseteq [n]$, we have
\begin{equation}\label{eq:S_identity_proof}
\sum_{I\subseteq[n]}(-1)^{\bigoplus_{l\in I}(a_{l}'\oplus a_l)}= \sum_{I\subseteq[n]}\prod_{l\in I}(-1)^{a_{l}'\oplus a_{l}}=\prod_{t=1}^{n}(1+(-1)^{a_{t}'\oplus a_{t}})    
\end{equation}
which establishes the required identity for $S(\mathbf{a}',\mathbf{a})$ once we observe that for even a single mismatch, i.e., $a_{t}'\neq a_t$, the product following the second equality in~\eqref{eq:S_identity_proof} is zero, and when $a_{t}'=a_t$ for all $t$, it is $2^n$.

Having derived the linear inverse in~\eqref{e:prob} expressing $p(\mathbf{a}\lvert\mathbf{x})$ in terms of its correlators, the $L^2$-projector sending any weakly-signalling behaviour $\mathbf f$ to its projection onto $\mathcal A$ is as shown below:
\begin{equation}\label{eq:ProjectionFormula}
\widehat{p}(\mathbf{a}\lvert\mathbf{x})
=\frac{1}{2^n}\sum_{I\subseteq[n]}\chi_{I}(\mathbf{a})\underbrace{\frac{1}{m^{n-|I|}}\sum_{\mathbf{x}_{\bar{I}},\mathbf{a}'}\chi_{I}(\mathbf{a}')f(\mathbf{a}'\lvert\mathbf{x})}_{\text{UMC (as in \eqref{eq:Marginal_Corr})}}.
\end{equation}
The formula in~\eqref{eq:ProjectionFormula} gives the projection because for an empirical behaviour satisfying the no-signalling conditions (and hence belonging to $\mathcal A$) the UMC is equal to the correlators defined in~\eqref{e:corr}, i.e.,
\begin{equation*}
\frac{1}{m^{n-\abs{I}}}\sum_{\mathbf{x}_{\bar I},\mathbf{a}}\chi_I(\mathbf{a})f(\mathbf{a}\lvert\mathbf{x}) = \sum_{\mathbf{a}}\chi_I(\mathbf{a})f(\mathbf{a}\lvert\mathbf{x}),
\end{equation*}
whenever $f(\mathbf{a}\lvert\mathbf{x})$ is strictly no-signalling, and the UMCs of $f$ and $\widehat p$ are the same by Corollary~\ref{cor:InvarianceUnderProjection_nm2}. 

The following commutative diagram illustrates the two approaches to performing the projection: a direct approach as formulated in~\eqref{eq:AffineProj_definition}, and the computationally simpler approach corresponding to the closed formula~\eqref{eq:ProjectionFormula} which can be decomposed into a sequence of three simple linear maps $T_3 T_2 T_1$.
\[ 
\begin{tikzcd}[scale=1.6, transform shape, row sep=4.1em, column sep=5.1em, every arrow/.append style = {-{Stealth[length=8pt]}, line width=1.05pt}]
\text{\large{$\mathbf{f}$}} \arrow[r, "T_1"] \arrow[d, "\Pi_{\mathcal A}"] & \text{\large{$C_{\mathbf x}^I$}} \arrow[d, "T_2"] \\%
\text{\large{$\widehat{\mathbf{p}}$}} & \text{\large{$\bar{C}_{\mathbf{x}_I}^I$}} \arrow[l, "T_3"]
\end{tikzcd}
\]
The three maps are as described below.

(1) \textit{Mapping $\mathbf f$ to $C_{\mathbf{x}}^I$}: Given a weakly-signalling empirical behaviour $\mathbf f$, the first summation in~\eqref{eq:Marginal_Corr} (which is over all combinations of $\mathbf{a}$) is:
\begin{equation}\label{eq:f_to_Cx}
C_{\mathbf{x}}^I = \sum_{\mathbf a}\chi_I(\mathbf{a})f(\mathbf{a}\lvert\mathbf{x}).
\end{equation}
Notice that the expression in~\eqref{eq:f_to_Cx} is the same as that in~\eqref{e:corr}: it is the sum of the (settings-conditional) probabilities of all outcomes, weighted by ``$+1$'' if the outcomes of the $I$-parties have even parity and by ``$-1$'' if they have odd parity. However, now $\mathbf f$ is weakly-signalling, and so~\eqref{eq:f_to_Cx} is not independent of the non-$I$-parties' input $\mathbf{x}_{\bar I}$, whereas, in~\eqref{e:corr}, the behaviour is no-signalling and so the expression is independent of $\mathbf{x}_{\bar I}$ and in all cases equal to the $C^I_{\mathbf{x}_I}$ quantity in \eqref{eq:nm2Corr}. We distinguish~\eqref{eq:f_to_Cx} notationally by using $C_{\mathbf x}^I$. The mapping corresponding to \eqref{eq:f_to_Cx} can be represented as $\mathbf{t}_1 = T_1\mathbf{f}$, where $T_{1}\in\mathbb{R}^{d\times d}$ for $d=(2m)^n$, and $\mathbf{t}_1$ collects the $C_{\mathbf{x}}^I$'s.

(2) \textit{Mapping $C_{\mathbf x}^I$ to $\bar{C}_{\mathbf{x}_I}^I$}: The second summation in~\eqref{eq:Marginal_Corr} uniformly averages $C_{\mathbf x}^I$ over all the input combinations $\mathbf{x}_{\bar I}$ of the parties not involved in $I$: 
\begin{equation}\label{eq:Cx_to_CxI}
\bar{C}_{\mathbf{x}_I}^I = \frac{1}{m^{n-\abs{I}}}\sum_{\mathbf{x}_{\bar{I}}}C_{\mathbf{x}}^I,
\end{equation}
which can be represented as $\mathbf{t}_2 = T_2\mathbf{t}_1$, where $T_2\in\mathbb{R}^{(m+1)^n\times d}$, and $\mathbf{t}_2$ collects the $\bar{C}_{\mathbf{x}_I}^I$'s.

(3) \textit{Mapping $\bar{C}_{\mathbf{x}_I}^I$ to $\widehat{\mathbf{p}}$}: Finally we map the UMC $\bar{C}_{\mathbf{x}_I}^I$ to the component probabilities of the projection $\widehat{\mathbf p}$ in $\mathcal A$.
\begin{equation}\label{eq:CxI_to_phat}
\widehat{p}(\mathbf{a}\lvert\mathbf{x}) = \frac{1}{2^n}\sum_{I\subseteq[n]}\chi_I(\mathbf{a})\bar{C}_{\mathbf{x}_I}^{I},
\end{equation}
which can be represented as $\widehat{\mathbf{p}} = T_3\mathbf{t}_2$, where $T_3\in\mathbb{R}^{d\times(m+1)^n}$. And so once we fix an ordering for the entries of $\mathbf{f},\mathbf{t}_1,\mathbf{t}_2,\widehat{\mathbf{p}}$, we can express the projection as $\widehat{\mathbf{p}}=T_3 T_2 T_1\mathbf{f}$, where the representation of the maps $T_i$ will depend on the ordering we choose to fix. An illustration of the maps is given in Appendix B for the $(2,2,2)$ Bell scenario; the explicit matrix forms there make the simplicity of the maps and their ease of implementation clear.

The diagram showing the two approaches for projection commutes because by Corollary~\ref{cor:InvarianceUnderProjection_nm2}, $\bar{C}_{\mathbf{x}_I}^I$ can be seen as either associated with the empirical behaviour $\mathbf f$ (from which it can be obtained after applying $T_1$ and $T_2$) or with its projected (no-signalling) counterpart $\widehat{\mathbf{p}}$ (from which it can be obtained using~\eqref{eq:nm2Corr}).

\subsection{Finite-sample probability for negative projected entries}\label{ss:FiniteSampleProbBound}
Although the affine projection $\Pi_{\mathcal A}$ restores the no-signalling and normalisation equalities, it does not explicitly impose the non-negativity constraints. We now give a finite-sample bound on the probability that a projected coordinate becomes negative, as a function of the smallest probability in the underlying behaviour. To that end,  we first re-express the projection formula given in~\eqref{eq:ProjectionFormula} in a form that will be useful for deriving probability bounds, as follows:
\begin{align}\label{eq:Re-epressProjFormula}
\widehat{p}(\mathbf{a}|\mathbf{x}) &= \sum_{\mathbf{a}',\mathbf{x}'}H_{\mathbf{a},\mathbf{x}}(\mathbf{a}',\mathbf{x}')f(\mathbf{a}',\mathbf{x}'),\\
\text{where }H_{\mathbf{a},\mathbf{x}}(\mathbf{a}',\mathbf{x}') &\coloneqq \frac{1}{(2m)^n}\prod_{i\in K(\mathbf{x},\mathbf{x}')}(1+m(-1)^{a_i\oplus a'_i}),\nonumber\\
\text{for }K(\mathbf{x},\mathbf{x}')&\coloneqq\{i\in[n]\colon x_i=x'_i\}.\nonumber
\end{align}
The details needed to obtain the above re-expression of~\eqref{eq:ProjectionFormula} are worked out in Appendix~\ref{a:NegativeProbIssue}.

Let $\mathbf{p}_0\in\mathcal{P}_{\mathrm{NS}}$ be the true behaviour, and $N_{\mathbf x}$ denote the independent trials corresponding to each setting combination $\mathbf x\in\mathcal{S}^n$. Since $\mathbf{p}_0\in\mathcal{A}$, the affine projection is unbiased:
\begin{equation*}
\mathbb{E}[\widehat{p}(\mathbf{a}|\mathbf{x})]=\mathbb{E}[\Pi_{\mathcal{A}}(f)(\mathbf{a}|\mathbf{x})] = \Pi_{\mathcal{A}}(\mathbb{E}[f(\mathbf{a}|\mathbf{x})])=\Pi_{\mathcal{A}}(p_0(\mathbf{a}|\mathbf{x})=p_0(\mathbf{a}|\mathbf{x})).
\end{equation*}
The justification for the penultimate equality is as follows. Suppose that for every setting $\mathbf x$ we perform $N_{\mathbf x}$ trials. For the $i$'th trial with setting $\mathbf x$, let $A_{\mathbf{x},i}\in\{0,1\}^n$ be the observed outcome string. We assume $A_{\mathbf{x},i}\sim p_{0}(\cdot|\mathbf{x})$, independently across trials. Then the empirical frequency can be written as $f(\mathbf{a}|\mathbf{x})=\frac{1}{N_{\mathbf x}}\sum_{i=1}^{N_{\mathbf x}}[\![A_{\mathbf{x},i}=\mathbf{a}]\!]$, where the function $[\![X]\!]$ evaluates to $1$ if the condition $X$ holds, and $0$ otherwise. This $f(\mathbf{a}|\mathbf{x})$ is a random variable because $A_{\mathbf{x},i}$ is random. Its expectation is
\begin{equation*}
\mathbb{E}[f(\mathbf{a}|\mathbf{x})] = \frac{1}{N_{\mathbf x}}\sum_{i=1}^{N_{\mathbf x}}\mathbb{E}[[\![A_{\mathbf{x},i}=\mathbf{a}]\!]]=\frac{1}{N_{\mathbf x}}\sum_{i=1}^{N_{\mathbf{x}}}\mathbb{P}[A_{\mathbf{x},i}=\mathbf{a}]= p_0(\mathbf{a}|\mathbf{x}).
\end{equation*}
So the empirical frequency is an unbiased estimator of the true probability. An important remark is that the statement $\mathbb{E}[f(\mathbf{a}|\mathbf{x})]=p_0(\mathbf{a}|\mathbf{x})$ does not mean that for one finite experiment $\widehat{p}(\mathbf{a}|\mathbf{x})=p_0(\mathbf{a}|\mathbf{x})$. Nonetheless, it does imply, via the law of large numbers, that if we repeat the whole finite-sample experiment many times, compute $\mathbf{f}$, project it to $\widehat{\mathbf{p}}$, and average the resulting $\widehat{p}(\mathbf{a}|\mathbf{x})$ values, the average will approach $p_0(\mathbf{a}|\mathbf{x})$. So $\widehat{p}(\mathbf{a}|\mathbf{x})$ fluctuates around $p_0(\mathbf{a}|\mathbf{x})$. The problem, then, is to bound the probability that the fluctuation dips below zero.

We now present the probability bound using Hoeffding/McDiarmid inequality~\cite{Hoeffding_63,McDiarmid_1989}. The following results are stated here in their final form; full derivations and implementation details can be found in Appendix~\ref{a:NegativeProbIssue}.

First, for a fixed $(\mathbf{a},\mathbf{x})$, we define the following quantity:
\begin{equation}
\Delta_{\mathbf{a},\mathbf{x}}(\mathbf{x}')\coloneqq \max_{\mathbf{a}'}H_{\mathbf{a},\mathbf{x}}(\mathbf{a}',\mathbf{x}') - \min_{\mathbf{a}'}H_{\mathbf{a},\mathbf{x}}(\mathbf{a}',\mathbf{x}'), 
\end{equation}
where $H_{\mathbf{a},\mathbf{x}}(\mathbf{a}',\mathbf{x}')$ is defined in~\eqref{eq:Re-epressProjFormula}. If $K(\mathbf{x},\mathbf{x}')\coloneqq\{i\colon x_i=x'_i\}$, and $k(\mathbf{x},\mathbf{x}')=|K(\mathbf{x},\mathbf{x}')|$, then as shown in Appendix~\ref{a:NegativeProbIssue}, we can write the following expression for $\Delta_{\mathbf{a},\mathbf{x}}(\mathbf{x}')$:
\begin{align}\label{eq:Delta_def}
\Delta_{\mathbf{a},\mathbf{x}}(\mathbf{x}') = \begin{cases}
    0, &k(\mathbf{x},\mathbf{x}')=0,\\
    \frac{2m(m+1)^{k(\mathbf{x},\mathbf{x}')-1}}{(2m)^n}, &k(\mathbf{x},\mathbf{x}')\ge 1.
\end{cases}
\end{align}
An application of Hoeffding/McDiarmid inequality for a fixed choice of $\mathbf a$ and $\mathbf x$ gives:
\begin{equation}\label{eq:McDiarmidIneq_SS3.1}
\mathbb{P}[\widehat{p}(\mathbf{a}|\mathbf{x})<0]=\mathbb{P}[\Pi_{\mathcal A}(f)(\mathbf{a}|\mathbf{x})<0] \le \exp\left(-\frac{2p_0(\mathbf{a}|\mathbf{x})^2}{\sum_{\mathbf{x}'}\Delta_{\mathbf{a},\mathbf{x}}(\mathbf{x}')^2/N_{\mathbf{x}'}}\right).
\end{equation}
Consequently, we can arrive at a union bound once we define the ``bad'' event coordinate wise as $\{\widehat{p}(\mathbf{a}|\mathbf{x})<0\}$. Then the event that the projected behaviour has at least one negative probability (coordinate) is $\{\widehat{\mathbf p}\notin\mathbb{R}_{+}^d\}$, which is equivalent to $\bigcup_{\mathbf{a},\mathbf{x}}\{\widehat{p}(\mathbf{a}|\mathbf{x})<0\}$. Therefore,
\begin{equation}\label{eq:ProbBound_UnbalancedCase}
\mathbb{P}[\widehat{\mathbf p}\notin\mathbb{R}_{+}^d] \le \sum_{\mathbf{a},\mathbf{x}}\exp\left(-\frac{2p_0(\mathbf{a}|\mathbf{x})^2}{\sum_{\mathbf{x}'}\Delta_{\mathbf{a},\mathbf{x}}(\mathbf{x}')^2/N_{\mathbf{x}'}}\right).
\end{equation}
Now defining $\gamma\coloneqq\min_{\mathbf{a},\mathbf{x}}p_0(\mathbf{a}|\mathbf{x})$, and assuming a balanced case where for each $\mathbf{x}$ we have $N_{\mathbf x}=N$,~\eqref{eq:ProbBound_UnbalancedCase} takes the form:
\begin{equation}\label{eq:ProbBound_BalancedCase}
\mathbb{P}[\widehat{\mathbf p}\notin\mathbb{R}_{+}^d] \le (2m)^n\exp\left(-\frac{2N\gamma^2}{C_{n,m}}\right),
\end{equation}
where $C_{n,m}$ is given by the formula $C_{n,m}=\frac{4m^2}{(2m)^{2n}(m+1)^2}[m^n(m+3)^n - (m-1)^n]$, derived in Appendix~\ref{a:NegativeProbIssue}. Applying this bound to the $(2,2,2)$ Bell scenario gives us $\mathbb{P}[\widehat{\mathbf p}\notin\mathbb{R}_{+}^d]\le 16\exp(-32N\gamma^2/11)$. While this bound can certainly be improved (since, for instance, the union bound ignores dependencies between the estimators of the coordinate-wise probabilities), it would be sufficient for many (though certainly not all) Bell experiment implementations to rule out the issue. For example, in a (2,2,2) Bell experiment with all underlying outcome probabilities at least 0.05---a reasonable assumption for an experiment aiming to observe CHSH values close to the Tsirelson bound---the bound tells us the probability of observing negative probabilities is already near-negligible ($<3\times10^{-4}$) with $N=1,500$ trials for each setting.

\section{Canonical Bell expressions}\label{s:CanonicalBellExpressions}
The kernel membership of the $k$-party UMC coefficient vectors shows that $L^2$-projection onto $\mathcal A$ preserves all correlator values. This leads to a class of robust Bell inequalities for the $(n,m,2)$ configuration whose evaluation is stable under projection of empirical behaviours onto the no-signalling affine hull. Concretely, given a Bell expression
\begin{equation}\label{e:bellfn}
\sum_{\mathbf a, \mathbf x} \gamma_{\mathbf a, \mathbf x} p(\mathbf a\lvert\mathbf x)
\end{equation}
(that is, a linear combination of components of $\mathbf p$ with fixed coefficients $\gamma_{\mathbf a, \mathbf x}$), there exist different coefficients $\{\beta_{I,\mathbf{x}_{I}}\}$ such that for $C_{\mathbf{x}_{I}}^{I}(\mathbf{p})$ as defined in~\eqref{eq:Marginal_Corr}, we have $B(\Pi_{\mathcal A}(\mathbf{p}))=B(\mathbf{p})$ for
\begin{equation}
B(\mathbf{p}) = \sum_{\mathbf{x}_{I}\colon I \subseteq [n]}\beta_{I,\mathbf{x}_{I}}\bar{C}_{\mathbf{x}_{I}}^{I}(\mathbf{p}).
\end{equation}
$B(\mathbf p)$ is equivalent to the original Bell expression in \eqref{e:bellfn} in the sense that the $B(\mathbf p) = \sum_{\mathbf a, \mathbf x} \gamma_{\mathbf a, \mathbf x} p(\mathbf a\lvert\mathbf x)$ for any no-signalling behaviour $\mathbf{p}$ (the expressions may differ when applied to a signalling behaviour). Thus, when this form of the Bell function is employed, the equivalence of pre- and post-projection Bell values mean that weakly-signalling artifacts from finite data do not affect the perceived nonlocality as indicated by a Bell inequality violation. This projection-invariant correlator expansion can thereby be taken as a canonical form of the Bell inequality. 

There are several Bell inequalities from different Bell scenarios that are classified as ``full correlator'' inequalities; two such examples are the so-called Clauser-Horne-Shimony-Holt (CHSH) inequality for the $(2,2,2)$ Bell scenario, and the Mermin inequality~\cite{Mermin_90} for the $(3,2,2)$ Bell scenario. These inequalities are already in the canonical form since the UMC terms in Section~\ref{s:CorrsFornm2BellScenario} for $I=[n]$ correspond to the standard definition of full correlators. The CHSH and Mermin inequalities are shown in~\eqref{eq:CHSH} and~\eqref{eq:Mermin}, respectively.
\begin{align}
&C_{00}^{\{1,2\}} + C_{01}^{\{1,2\}} + C_{10}^{\{1,2\}} - C_{11}^{\{1,2\}} \le 2,\label{eq:CHSH}\\
&C_{001}^{\{1,2,3\}} + C_{010}^{\{1,2,3\}} + C_{100}^{\{1,2,3\}} - C_{111}^{\{1,2,3\}} \le 2.\label{eq:Mermin}
\end{align}
The full correlators for the two inequalities are defined for $a,b,c,x,y,z\in\{0,1\}$ as:
\begin{align}
&C_{xy}^{\{1,2\}}\coloneqq \sum_{a,b}(-1)^{a\oplus b}p(ab\lvert xy),\label{eq:FullCorrelator_CHSH}\\
&C_{xyz}^{\{1,2,3\}}\coloneqq\sum_{a,b,c}(-1)^{a\oplus b\oplus c}p(abc\lvert xyz).\label{eq:FullCorrelator_Mermin}
\end{align}

The canonicalisation step is non-trivial precisely when marginal correlator terms appear in the Bell inequality; full-correlator inequalities need no further rewriting. We illustrate the passage to canonical form using some Bell and Bell-like inequalities that serve as linear witnesses for quantum phenomenon such as nonlocality and local operations and shared randomness (LOSR)-based genuine tripartite nonlocality.

\subsection{Family of tilted CHSH inequalities}
We consider first the tilted Bell inequality from~\cite{AMP2012}:
\begin{equation}\label{eq:AcinMassarPironio_BellIneq}
\beta C_{0}^{\{1\}}+\sum_{x,y=0}^{1}\alpha^{1-x}(-1)^{xy}C_{xy}^{\{1,2\}}\le\beta+2\alpha.
\end{equation}
In~\eqref{eq:AcinMassarPironio_BellIneq}, $\alpha\ge 1,\beta\ge 0$. For $\beta=0$ and $\alpha=1$, it corresponds to the CHSH inequality. The marginal correlator in the Bell functional in~\eqref{eq:AcinMassarPironio_BellIneq} when defined assuming the no-signalling constraints to hold has the following form, following~\eqref{e:corr}:
\begin{equation}
C_{x}^{\{1\}}\coloneqq\sum_{a,b=0}^{1}(-1)^{a}p(ab\lvert xy).\label{eq:AMP_MarginalCorr}
\end{equation}
The full correlator is as given in~\eqref{eq:FullCorrelator_CHSH}. While the full correlator stays unaffected regardless of signalling, the marginal correlator is well-defined only if the no-signalling constraints hold for the behaviour $\mathbf p$, since the following equality is satisfied, as a consequence, for $y\neq y'$:
\begin{equation}\label{eq:C_x_WellDefined}
\sum_{a=0}^{1}(-1)^{a}p(a\lvert xy)=\sum_{a=0}^{1}(-1)^{a}p(a\lvert xy'),
\end{equation}
where $p(a\lvert x,y)=\sum_{b}p(ab\lvert xy)$. For a weakly-signalling behaviour, however, the equality in~\eqref{eq:C_x_WellDefined} need not hold, so \textit{a priori} it is not clear how to evaluate the Bell function of \eqref{eq:AcinMassarPironio_BellIneq} on a signalling behaviour $\mathbf{p}$. To proceed, a version of the Bell functional in~\eqref{eq:AcinMassarPironio_BellIneq} that is invariant under $L^2$-projection is obtained by replacing the marginal correlator $C_{0}^{\{1\}}$ with its projection-invariant version (the uniformly-averaged marginal correlator), as defined in~\eqref{eq:Marginal_Corr}. We note this procedure can be observed in an earlier work in the (2,2,2) case; see~\cite{Nieto-Silleras_2018} equations (37)-(40) and footnote 2. Plugging the UMC, which is $\frac{1}{2}\sum_{y}\sum_{a,b}(-1)^{a}p(ab\lvert 0y)$, in~\eqref{eq:AcinMassarPironio_BellIneq} and writing out the full correlators $C_{xy}^{\{1,2\}}$ as linear combinations of the conditional probabilities (as shown in~\eqref{eq:FullCorrelator_CHSH}), we have the projection-invariant version, shown in~\eqref{eq:AMP_ProjInvVersion}, which can be considered as a canonical version. In~\eqref{eq:AMP_ProjInvVersion}, $\mathbf{p}_{xy}\in\mathbb{R}^4$ denotes the vector $(p(00\lvert xy),p(01\lvert xy),p(10\lvert xy),p(11\lvert xy))^T$, for $x,y\in\{0,1\}$, and $\mathbf{w}_{\alpha,\beta}=(\beta\mathbf{r}+2\alpha\mathbf{s})/2$ for $\mathbf{r}=(1,1,-1,-1)^T, \mathbf{s}=(1,-1,-1,1)^T$.
\begin{equation}\label{eq:AMP_ProjInvVersion}
\mathbf{w}_{\alpha,\beta}^T(\mathbf{p}_{00}+\mathbf{p}_{01})+\mathbf{s}^T(\mathbf{p}_{10}-\mathbf{p}_{11})\le\beta+2\alpha.
\end{equation}

\subsection{\texorpdfstring{$I_{3322}$}{I3322} inequality}
The $I_{3322}$ inequality is one of the two classes of non-trivial facet inequalities of the local polytope of behaviours for the $(2,3,2)$ Bell scenario~\cite{Sliwa_2003}; the other class being the liftings of CHSH. Its standard representation is a linear combination of conditional probabilities (as shown, for instance, in expression (19) in~\cite{DC_2004}):
\begin{align}\label{eq:I3322_probForm}
-p_{1}(0\lvert 0) - 2p_{2}(0\lvert 0) - p_{2}(0\lvert 1) + p(00\lvert 00)&\nonumber \\
+\,p(00\lvert 10) + p(00\lvert 20) + p(00\lvert 01) + p(00\lvert 11)&\nonumber \\
- p(00\lvert 21) + p(00\lvert 02) - p(00\lvert 12) &\le 0.
\end{align}
In~\eqref{eq:I3322_probForm}, for $a,b\in\{0,1\}$ and $x,y\in\{0,1,2\}$, $p_{1}(a\lvert x)$ and $p_{2}(b\lvert y)$ are the marginal conditional probabilities of the two parties $\mathsf{A}_1$ and $\mathsf{A}_2$, respectively, which are well-defined only when no-signalling holds for (the joint) $p(ab\lvert xy)$. It can be expressed equivalently in terms of correlators using the definitions in~\eqref{eq:FullCorrelator_CHSH} and~\eqref{eq:AMP_MarginalCorr} as
\begin{align}\label{eq:I3322_corrForm}
C_{0}^{\{1\}} + C_{1}^{\{1\}} - C_{0}^{\{2\}} - C_{1}^{\{2\}} + C_{00}^{\{1,2\}} &\nonumber \\
+\, C_{01}^{\{1,2\}} + C_{10}^{\{1,2\}} + C_{11}^{\{1,2\}} + C_{20}^{\{1,2\}}&\nonumber \\
-\, C_{21}^{\{1,2\}} + C_{02}^{\{1,2\}} - C_{12}^{\{1,2\}} &\le 4,
\end{align}
which can be transformed to the canonical form by performing a like-for-like substitution---replacing the standard correlators by their projection-invariant (UMC) counterparts. The canonical form for~\eqref{eq:I3322_probForm} is then the following:
\begin{equation}\label{eq:I3322_canonicalForm}
\mathbf{u}^T\sum_{x,y=0}^{1}\mathbf{p}_{xy} + \mathbf{v}^T(\mathbf{p}_{20}-\mathbf{p}_{21}+\mathbf{p}_{02}-\mathbf{p}_{12}) \le 4,    
\end{equation}
where $\mathbf{u}=(1,0,-2,1)^T, \mathbf{v}=(1,-1,-1,1)^T$, and $\mathbf{p}_{xy}\equiv(p(00\lvert xy),p(01\lvert xy),p(10\lvert xy),p(11\lvert xy))^T$.

\subsection{Bell-like inequality for LOSR-GTNL}
A Bell-like inequality serving as a linear witness for LOSR-based genuine tripartite nonlocality~\cite{Mao2022} is shown below:
\begin{equation}\label{eq:MaoLinearWitness}
C_{00}^{\{1,2\}} + C_{01}^{\{1,2\}} + C_{101}^{\{1,2,3\}} - C_{111}^{\{1,2,3\}} + 2C_{00}^{\{1,3\}} \le 4,
\end{equation}
where the marginal correlator terms are defined below for $a,b,c,x,y,z\in\{0,1\}$. In~\eqref{eq:MaoLinearWitness}, the marginal correlators $C_{xy}^{\{1,2\}}$ and $C_{xz}^{\{1,3\}}$ are defined analogously to~\eqref{eq:FullCorrelator_CHSH} and the full correlators are as defined in~\eqref{eq:FullCorrelator_Mermin}. Replacing each correlator term in~\eqref{eq:MaoLinearWitness} with the corresponding UMC term gives the canonical version of the inequality shown below:
\begin{equation}\label{eq:MaoLinearWitness_Canonical}
\frac{1}{2}\mathbf{h}^T(\mathbf{p}_{001}+\mathbf{p}_{011})+\mathbf{t}^T(\mathbf{p}_{101}-\mathbf{p}_{111})
+\,\frac{1}{2}(\mathbf{h}+2\mathbf{r})^T(\mathbf{p}_{000}+\mathbf{p}_{010})\le 4,
\end{equation}
where $\mathbf{p}_{xyz}\equiv\big(p(abc\lvert xyz)\big)_{abc}^T\in\mathbb{R}^8$ for $a,b,c\in\{0,1\}$ arranged in lexicographic order with $a$ as the most-significant and $c$ as the least-significant bit. The ``parity'' weight vectors over outcomes $\mathbf{h},\mathbf{r},\mathbf{t}\in\mathbb{R}^8$ are defined as: $\mathbf{h}\coloneqq\big((-1)^{a\oplus b}\big)_{abc}^T$, $\mathbf{r}\coloneqq\big((-1)^{a\oplus c}\big)_{abc}^T$ and $\mathbf{t}\coloneqq\big((-1)^{a\oplus b\oplus c}\big)_{abc}^T$. Equivalently, they can be seen as $\mathbf{h}=(\mathbf{a}\otimes\mathbf{a}\otimes\mathbf{1}_2)^T$, $\mathbf{r}=(\mathbf{a}\otimes\mathbf{1}_2\otimes\mathbf{a})^T$ and $\mathbf{t}=(\mathbf{a}^{\otimes 3})^T$ for $\mathbf{a}=(1,-1)$ and $\mathbf{1}_2=(1,1)$.

\subsection{Scope and nonlinear Bell-type inequalities}
The canonicalisation procedure discussed before (in the earlier subsections) and its mechanism is linear algebraic. A linear Bell expression is not uniquely represented outside the no-signalling set. Indeed, two linear expressions may differ by a linear combination of no-signalling constraints, and hence are equivalent on $\mathcal{P}_{\mathrm{NS}}$, while giving different values on weakly signalling empirical behaviours for which those constraints are not exactly satisfied. The UMC/correlator representation selects a representative of this equivalence class whose coefficient vector lies in $\mathrm{ker}(A_{\mathrm{eq}})$, and is therefore invariant under the affine projection $\Pi_{\mathcal A}$.

This argument does not directly extend to arbitrary nonlinear Bell-type inequalities. If a nonlinear witness can be written as $F(\{C_{\mathbf{x}_I}^I\}_{\mathbf{x}_I,I})$, where $F$ depends on projection-invariant full-correlators and UMCs, then its value is also invariant under $\Pi_{\mathcal A}$, since each argument of $F$ is invariant. However, a generic nonlinear expression in the raw probabilities $p(\mathbf{a}|\mathbf{x})$, or in marginal quantities not replaced by their UMC counterparts, need not be invariant under the affine projection. Thus the results of this work should be understood as giving a canonical projection-invariant representative for linear Bell inequalities in the $(n,m,2)$ scenario. Nonlinear Bell-type inequalities, including those arising in network or entropic formulations, require a separate analysis.

\section{Generalisation to weighted \texorpdfstring{$L^2$}{L2} norms}
Our work on $L^2$ projections up to this point is most naturally tailored to a scenario under which all settings occur with equal probability, such as in Table~\ref{tab:Distr_Sig_222} where each of the four measurement configurations is sampled about 1.25~M times.  This is common for most implementations of Bell experiments, but one can consider situations where settings occur with different probabilities. Indeed, some quantum information theoretic protocols motivate non-uniform settings distributions, such as device-independent quantum key distribution~\cite{Acín_2006,Arnon-Friedman2018,Schwonnek2021}, where secret key is extracted from a single setting configuration while other configurations are sampled only occasionally to perform ``spot checks'' to ensure the protocol remains secure, and device-independent randomness expansion~\cite{Miller2017,Bhavsar_2023,Shalm2021,Liu2021}, where randomness for selecting settings is considered an input resource to be minimised through utilisation of a non-uniform settings distribution.

Referring to Table \ref{tab:Distr_Sig_222}, if there were instead (say) $1$ million counts for the first row corresponding to setting $xy=00$, but only about ~1,000 counts for each of the other three setting configurations, one would expect empirically computed probabilities of the form $p(\cdot|00)$ to exhibit smaller deviations from the true underlying distribution compared to empirical estimates of $p(\cdot\lvert 01)$, $p(\cdot\lvert 10)$, $p(\cdot\lvert 11)$. Correspondingly, when projecting onto the no-signalling set, it would make sense to modify the least squares weighting such that deviations in $p(\cdot|00)$ terms are penalised more heavily, in proportion to the higher probability with which these settings occur. 

This is naturally formalised by encoding the different weightings with a generalised inner product $\langle \cdot, \cdot \rangle_D$ given by the expression  
\begin{equation}\label{eq:genip}
    \langle \mathbf p_1,\mathbf p_2 \rangle_D = \mathbf p_1^T D\mathbf p_2
\end{equation}
where $D\in\mathbb{R}^{d\times d}$ is a diagonal matrix with positive diagonal entries reflecting the different settings weights; \eqref{eq:genip} defines an inner product in this case because $D$ is positive definite. The corresponding norm $\norm{\cdot}_D$ is induced by the expression
\begin{equation}\label{eq:gennorm}
\norm{\mathbf p}_D= \sqrt{\langle \mathbf p, \mathbf p\rangle_D}
\end{equation}
according to which computing the weighted projection of $\mathbf f$ onto the no-signalling space now becomes the question of finding
\begin{equation}\label{eq:genLS}
\widehat{\mathbf p}=\argmin_{\mathbf{p}\in\mathcal{A}}\norm{\mathbf{f}-\mathbf{p}}_{D}^{2}.
\end{equation}
The solution to \eqref{eq:genLS} is given by 
\begin{equation}\label{eq:genAffProj}
\Pi^D_{\mathcal{A}}(\mathbf{f})\coloneqq\Pi^D_{\mathrm{ker}}(\mathbf{f}-\mathbf{d})+\mathbf{d}
\end{equation}
where $\mathbf d$ is a displacement vector in $\mathcal A$ and $\Pi^D_{\mathrm{ker}}$ can be obtained by solving the generalised normal equations for minimum $\norm{\cdot}_D$-distance to the subspace $\mathrm{ker}(A_{\mathrm{eq}})$, which yields the expression
\begin{equation}
\Pi^D_{\mathrm{ker}}=B(B^TDB)^{-1}B^TD,
\end{equation}
recalling $B$ is a matrix whose columns form a basis of $\mathrm{ker}(A_{\mathrm{eq}})$. 

UMCs are no longer invariant under the affine projection of \eqref{eq:genAffProj}, but averaged correlators with weightings corresponding to the $D$ entries are. Specifically, terms of the form $\langle\mathbf{c}_{\tilde{\mathbf{u}}_I}^I,\mathbf{p}\rangle_D$ are invariant, but now the original (uniform) $\mathbf{c}_{\tilde{\mathbf{u}}_I}^I$ coefficients yield a different averaging via the $D$-weighted inner product. To see this invariance, take the displacement vector in~\eqref{eq:genAffProj} to be $\mathbf{d}'$ such that $\Pi_{\mathrm{ker}}^D\mathbf{d}'=\mathbf{0}$, which can be obtained from any $\mathbf d \in \mathcal A$ by setting $\mathbf{d}' = \mathbf d - \Pi^D_{\mathrm{ker}}\mathbf d$. This replacement will also be orthogonal (with respect to $\langle \cdot,\cdot\rangle_D$) to any vector in $\mathbf v \in \mathrm{ker}(A_\mathrm{eq})$, because noting $\mathbf v$ is expressible as a linear combination of $B$ columns $\mathbf v = B\mathbf x$, we have
\begin{equation*}
\langle\mathbf{d}',\mathbf{v}\rangle_D
= \mathbf{d}^T D\mathbf{v} - \mathbf{d}^T\big(B(B^T DB)^{-1}B^T D\big)^T D\mathbf{v}
= \mathbf{d}^T D\mathbf{v} - \mathbf{d}^T DB(B^T DB)^{-1}B^T D\mathbf{v}=\mathbf 0.
\end{equation*}
And so with this choice of
displacement vector, the definition in~\eqref{eq:genAffProj} becomes $\Pi_{\mathrm{ker}}^D\mathbf{f}+\mathbf{d}'$. Now, since $\mathbf{c}_{\tilde{\mathbf{u}}_I}^I$ is still in $\mathrm{ker}(A_{\mathrm{eq}})$---the linear structure of the vector space is unchanged by adopting a new inner product---so, using shorthand $\mathbf c = \mathbf{c}_{\tilde{\mathbf{u}}_I}^I$, we can write
\begin{align}\label{eq:genproof}
\langle\mathbf{c},\widehat{\mathbf{p}}\rangle_D &= \langle\mathbf{c},\Pi^D_{\mathrm{ker}}\mathbf{f}\rangle_D +\langle\mathbf{c},\mathbf{d}'\rangle_D = \mathbf c^TDB(B^TDB)^{-1}B^TD\mathbf f \nonumber\\
&=\big(B(B^TDB)^{-1}B^TD\mathbf c\big)^TD\mathbf f = \mathbf c^TD\mathbf f=\langle\mathbf{c},\mathbf{f}\rangle_D,
\end{align}
where the invariance of $\mathbf c$ under the map $B(B^TDB)^{-1}B^TD$ is seen by writing $\mathbf c$ as $B\mathbf x$.

We remark that while the $D$-projection-invariant quantity $\langle \mathbf{c}_{\tilde{\mathbf{u}}_I}^I, \cdot \rangle_D$ is a non-uniform average of correlators, for a \textit{no-signalling} behaviour, all the averaged terms are equal to the standard no-signalling correlator given in \eqref{eq:nm2Corr}, and so their average will also equal \eqref{eq:nm2Corr}. This means that the results of Section \ref{s:ClosedFormProj} can be immediately adapted to generate a simple closed-form expression for the $D$-weighted projection by changing (only) the map $T_1$ to account for the non-uniform averaging, thus allowing the corresponding computational simplification for obtaining the closest no-signalling approximation to an empirical distribution in a setup with unbalanced setting weights. The results of Section \ref{s:CanonicalBellExpressions} can be generalised as well to construct weighted-$L_2$-projection invariant Bell inequalities, though the deeper significance of uniquely capturing no-signalling nonlocality is best expressed through the balanced canonical forms of Section \ref{s:CanonicalBellExpressions}.

\section{Conclusion}
We presented a projection-invariant formulation of linear Bell inequalities for 
$(n,m,2)$ Bell scenarios and a sparse, closed formula for the projector that maps empirical frequencies to the no-signalling affine hull without altering the value of a canonical Bell functional. The key observation is that full correlators and uniformly-averaged marginal correlators lie in the kernel of the matrix encoding the no-signalling and normalisation constraints, so evaluating inequalities in this correlator basis makes the reported violation insensitive to weak signalling in finite data when the canonical form is used. Computationally, our three-map pipeline implements the projection with structured sparsity, avoiding dense inverses and scaling cleanly to larger instances, and it can be readily generalised to applications in which differently weighted projections are desired. This canonical, project-then-evaluate workflow standardises comparisons across experiments and benefits device-independent tasks such as randomness generation, QKD, and multipartite nonlocality/entanglement certification, by ensuring that the Bell value reflects nonlocal structure rather than residual signalling.
 
\section*{Acknowledgements}
This work was partially supported by NSF Award Nos.~$2210399$ and~$2328800$.

\appendix

\section{Probability bound for the event of negative entries in \texorpdfstring{$\widehat{\mathbf p}$}{phat}}\label{a:NegativeProbIssue}
Here we furnish the relevant details related to the results presented in Section~\ref{ss:FiniteSampleProbBound}. We begin with the derivation of the re-expression in~\eqref{eq:Re-epressProjFormula} of the projection formula in~\eqref{eq:ProjectionFormula}. Starting with the formula given in~\eqref{eq:ProjectionFormula}, i.e.,
\begin{equation}\label{eq:ProjFormula1}
\widehat{p}(\mathbf{a}|\mathbf{x}) = \frac{1}{2^n}\sum_{I\subseteq[n]}\chi_I(\mathbf{a})\frac{1}{m^{n-|I|}}\sum_{\mathbf{x}_{\bar I},\mathbf{a}'}\chi_I(\mathbf{a}')f(\mathbf{a}'|\mathbf{x}_I,\mathbf{x}_{\bar{I}}),
\end{equation}
where $\mathbf{x}_I$ is fixed to equal the values of the target input string $\mathbf{x}$ on the subset $I$, and making explicit that in the inner sum the input string $(\mathbf{x}_I,\mathbf{x}_{\bar{I}})$ agrees with the target input $\mathbf x$ on the coordinates in $I$ while the coordinates in $\bar{I}$ are averaged over, we now introduce a full input string $\mathbf{x}'\in\mathcal{S}^n$, where the condition that $\mathbf{x}'$ agrees with $\mathbf{x}$ on $I$ is $x_i=x'_i$, for all $i\in I$. Therefore,
\begin{equation*}
\sum_{\mathbf{x}_{\bar I}}f(\mathbf{a}'|\mathbf{x}_I,\mathbf{x}_{\bar I}) = \sum_{\mathbf{x}'}[\![I\subseteq K(\mathbf{x},\mathbf{x}')]\!]f(\mathbf{a}'|\mathbf{x}'),
\end{equation*}
where the function $[\![X]\!]$ evaluates to $1$ if the condition $X$ holds, and $0$ otherwise. Substituting the above formula into~\eqref{eq:ProjFormula1} we get:
\begin{align*}
\widehat{p}(\mathbf{a}|\mathbf{x}) &= \sum_{\mathbf{a}',\mathbf{x}'}\Big[\frac{1}{2^n}\sum_{I\subseteq[n]}\frac{[\![I\subseteq K(\mathbf{x},\mathbf{x}')]\!]}{m^{n-|I|}}\chi_I(\mathbf{a})\chi_I(\mathbf{a}')\Big]f(\mathbf{a}'|\mathbf{x}')\\
&= \sum_{\mathbf{a}',\mathbf{x}'}\Big[\frac{1}{2^n}\sum_{I\subseteq K(\mathbf{x},\mathbf{x}')}\frac{1}{m^{n-|I|}}\chi_I(\mathbf{a})\chi_I(\mathbf{a}')\Big]f(\mathbf{a}'|\mathbf{x}')\\
&= \sum_{\mathbf{a}',\mathbf{x}'}\Big[\frac{1}{2^n}\sum_{I\subseteq K(\mathbf{x},\mathbf{x}')}\frac{1}{m^{n-|I|}}\prod_{i\in I}(-1)^{a_i\oplus a'_i}\Big]f(\mathbf{a}'|\mathbf{x}')\\
&= \sum_{\mathbf{a}',\mathbf{x}'}\Big[\underbrace{\frac{1}{(2m)^n}\sum_{I\subseteq K(\mathbf{x},\mathbf{x}')}\prod_{i\in I}m(-1)^{a_i\oplus a'_i}}_{\eqqcolon H_{\mathbf{a},\mathbf{x}}(\mathbf{a}',\mathbf{x}')}\Big]f(\mathbf{a}'|\mathbf{x}')
\end{align*}
Using the elementary product identity $\sum_{I\subseteq K}\prod_{i\in I}t_i = \prod_{i\in K}(1+t_i)$ with $t_i = m(-1)^{a_i\oplus a'_i}$ in the expression of $H_{\mathbf{a},\mathbf{x}}(\mathbf{a}',\mathbf{x}')$ above, we can write 
\begin{align*}
\widehat{p}(\mathbf{a}|\mathbf{x}) &= \sum_{\mathbf{a}',\mathbf{x}'}H_{\mathbf{a},\mathbf{x}}(\mathbf{a}',\mathbf{x}')f(\mathbf{a}'|\mathbf{x}'),\\
\text{where }H_{\mathbf{a},\mathbf{x}}(\mathbf{a}',\mathbf{x}') &= \frac{1}{(2m)^n}\prod_{i\in K(\mathbf{x},\mathbf{x}')}(1+m(-1)^{a_i\oplus a'_i}).
\end{align*}

Having derived the alternate formula for $\widehat{p}(\mathbf{a}|\mathbf{x})$ above, we now can utilise it to apply a probability bound. The applicability condition of the finite-sample probability-bound inequality that we are trying to apply is referred to as the bounded differences property: Let $Z_1,Z_2,\ldots,Z_M$ be independent random variables, and let $F=F(Z_1,Z_2,\ldots,Z_M)$ be a real-valued function. Suppose changing only the $j$'th input can change $F$ by at most $c_j$, i.e.,
\begin{equation*}
\sup_{z_1,\ldots,z_M,z'_j}|F(z_1,\ldots,z_j,\ldots,z_M)-F(z_1,\ldots,z'_j,\ldots,z_M)| \le c_j.
\end{equation*}
Then the McDiarmid inequality says~\cite{McDiarmid_1989,Hoeffding_63}: 
\begin{equation}\label{eq:McDiarmidIneq}
\mathbb{P}[F-\mathbb{E}[F]\le -t] \le \exp\left(-\frac{2t^2}{\sum_{j=1}^M c_j^2}\right).
\end{equation}
For our purposes, we only need the one-sided version of the inequality given above. Now what are the $Z_j$'s in our problem? Fix a coordinate $(\mathbf{a},\mathbf{x})$. The random variables are the observed outcomes in the finite experiment. For each setting $\mathbf{x}'$, suppose we perform $N_{\mathbf{x}'}$ trials. Let $A_{\mathbf{x}',r}\in\{0,1\}^n$, for $r=1,2,\ldots,N_{\mathbf{x}'}$, be the outcome-string in the $r$'th trial with setting $\mathbf{x}'$. Under the sampling model, we have $A_{\mathbf{x}',r}\sim p_0(\cdot|\mathbf{x}')$, independent over $\mathbf{x}'$ and $r$. The empirical frequencies are:
\begin{equation*}
f(\mathbf{a}'|\mathbf{x}')=\frac{1}{N_{\mathbf{x}'}}\sum_{r=1}^{N_{\mathbf{x}'}}[\![A_{\mathbf{x}',r}=\mathbf{a}']\!].
\end{equation*}
Since $\widehat{p}(\mathbf{a}|\mathbf{x})$ is a linear function of $f$, it is also a function of all the random variables $A_{\mathbf{x}',r}$.

Next, for a fixed $(\mathbf{a},\mathbf{x})$, define $F\coloneqq \widehat{p}(\mathbf{a}|\mathbf{x})$. Substituting the expression of the empirical frequencies above in~\eqref{eq:Re-epressProjFormula} we get:
\begin{equation*}
\sum_{\mathbf{a}',\mathbf{x}'}H_{\mathbf{a},\mathbf{x}}(\mathbf{a}',\mathbf{x}')f(\mathbf{a}',\mathbf{x}') = \sum_{\mathbf{a}',\mathbf{x}'}H_{\mathbf{a},\mathbf{x}}(\mathbf{a}',\mathbf{x}')\frac{1}{N_{\mathbf{x}'}}\sum_{r=1}^{N_{\mathbf{x}'}}[\![A_{\mathbf{x}',r}=\mathbf{a}']\!]=\sum_{\mathbf{x}'}\frac{1}{N_{\mathbf{x}'}}\sum_{r=1}^{N_{\mathbf{x}'}}H_{\mathbf{a},\mathbf{x}}(A_{\mathbf{x}',r},\mathbf{x}').
\end{equation*}
Thus,
\begin{equation*}
F = F(\{A_{\mathbf{x}',r}\}_{\mathbf{x}',r}) = \sum_{\mathbf{x}'}\frac{1}{N_{\mathbf{x}'}}\sum_{r=1}^{N_{\mathbf{x}'}}H_{\mathbf{a},\mathbf{x}}(A_{\mathbf{x}',r},\mathbf{x}'),
\end{equation*}
which is actually an independent sum of bounded random variables, so one could apply Hoeffding inequality directly. McDiarmid inequality gives the same type of bound. Now we ask the question: How much can $F=\widehat{p}(\mathbf{a}|\mathbf{x})$ change if we change one observed outcome? Fix a setting $\mathbf{x}'$ and a trial $r$. Suppose the outcome changes from $A_{\mathbf{x}',r}=\mathbf{a}'$ to $A_{\mathbf{x}',r}=\mathbf{a}''$. The change $|\Delta F|$ in $F$ is
\begin{equation*}
\Big|\frac{1}{N_{\mathbf{x}'}}H_{\mathbf{a},\mathbf{x}}(\mathbf{a}',\mathbf{x}')-H_{\mathbf{a},\mathbf{x}}(\mathbf{a}'',\mathbf{x}')\Big|
\end{equation*}
Therefore,
\begin{equation}\label{eq:Change_in_F_atMost_Delta}
|\Delta F| \le \frac{1}{N_{\mathbf{x}'}}\Big(\max_{\mathbf{a}'}H_{\mathbf{a},\mathbf{x}}(\mathbf{a}',\mathbf{x}')-\min_{\mathbf{a}'}H_{\mathbf{a},\mathbf{x}}(\mathbf{a}',\mathbf{x}')\Big)
\end{equation}
Define
\begin{equation}\label{eq:Range}
\Delta_{\mathbf{a},\mathbf{x}}(\mathbf{x}')\coloneqq\max_{\mathbf{a}'}H_{\mathbf{a},\mathbf{x}}(\mathbf{a}',\mathbf{x}')-\min_{\mathbf{a}'}H_{\mathbf{a},\mathbf{x}}(\mathbf{a}',\mathbf{x}')
\end{equation}
Then changing one outcome in the $\mathbf{x}'$-block changes $F$ by at most $c_{\mathbf{x}',r}=\Delta_{\mathbf{a},\mathbf{x}}(\mathbf{x}')/N_{\mathbf{x}'}$. There are $N_{\mathbf{x}'}$ such trials in the $\mathbf{x}'$-block, so the denominator in the McDiarmid bound in~\eqref{eq:McDiarmidIneq} becomes
\begin{equation*}
\sum_{\mathbf{x}'}\sum_{r=1}^{N_{\mathbf{x}'}}c_{\mathbf{x}',r}^2 = \sum_{\mathbf{x}'}N_{\mathbf{x}'}\left(\frac{\Delta_{\mathbf{a},\mathbf{x}}(\mathbf{x}')}{N_{\mathbf{x}'}}\right)^2 = \sum_{\mathbf{x}'}\frac{\Delta_{\mathbf{a},\mathbf{x}}(\mathbf{x}')^2}{N_{\mathbf{x}'}}.
\end{equation*}
Now we can apply the inequality which gives us:
\begin{equation*}
\mathbb{P}[\widehat{p}(\mathbf{a}|\mathbf{x})-\mathbb{E}[\widehat{p}(\mathbf{a}|\mathbf{x})] \le -t] \le \exp\left(-\frac{2t^2}{\sum_{\mathbf{x}'}\Delta_{\mathbf{a},\mathbf{x}}(\mathbf{x}')^2/N_{\mathbf{x}'}}\right).
\end{equation*}
Since $\mathbb{E}[\widehat{p}(\mathbf{a}|\mathbf{x})]=p_0(\mathbf{a}|\mathbf{x})$, as discussed in Section~\ref{ss:FiniteSampleProbBound}, we can rewrite the above as:
\begin{equation*}
\mathbb{P}[\widehat{p}(\mathbf{a}|\mathbf{x})-p_0(\mathbf{a}|\mathbf{x}) \le -t] \le \exp\left(-\frac{2t^2}{\sum_{\mathbf{x}'}\Delta_{\mathbf{a},\mathbf{x}}(\mathbf{x}')^2/N_{\mathbf{x}'}}\right).
\end{equation*}
To bound negativity, we set $t=p_0(\mathbf{a}|\mathbf{x})$. Then $\widehat{p}(\mathbf{a}|\mathbf{x})<0\Leftrightarrow\widehat{p}(\mathbf{a}|\mathbf{x})-p_0(\mathbf{a}|\mathbf{x})<-p_0(\mathbf{a}|\mathbf{x})$. And so the above can be re-expressed further as:
\begin{equation*}
\mathbb{P}[\widehat{p}(\mathbf{a}|\mathbf{x}) < 0] \le \exp\left(-\frac{2p_0(\mathbf{a}|\mathbf{x})^2}{\sum_{\mathbf{x}'}\Delta_{\mathbf{a},\mathbf{x}}(\mathbf{x}')^2/N_{\mathbf{x}'}}\right),
\end{equation*}
which is what we have in~\eqref{eq:McDiarmidIneq_SS3.1}.

We now turn to deriving the expression for $\Delta_{\mathbf{a},\mathbf{x}}(\mathbf{x}')$ in~\eqref{eq:Delta_def}. Fix $\mathbf{a},\mathbf{x},\mathbf{x}'$, and recall that
\begin{equation*}
H_{\mathbf{a},\mathbf{x}}(\mathbf{a}',\mathbf{x}') = \frac{1}{(2m)^n}\prod_{i\in K(\mathbf{x},\mathbf{x}')}(1+m(-1)^{a_i\oplus a'_i}),
\end{equation*}
where $K(\mathbf{x},\mathbf{x}')\coloneqq\{i\colon x_i=x'_i\}$ and $k(\mathbf{x},\mathbf{x}')=|K(\mathbf{x},\mathbf{x}')|$. The range, as defined in~\eqref{eq:Range}, is 
\begin{equation*}
\Delta_{\mathbf{a},\mathbf{x}}=\max_{\mathbf{a}'}H_{\mathbf{a},\mathbf{x}}(\mathbf{a}',\mathbf{x}')-\min_{\mathbf{a}'}H_{\mathbf{a},\mathbf{x}}(\mathbf{a}',\mathbf{x}').
\end{equation*}
We first look at the case when $K(\mathbf{x},\mathbf{x}')=\varnothing$. Then $k(\mathbf{x},\mathbf{x}')=|K(\mathbf{x},\mathbf{x}')|=0$, and the product in the expression of $H_{\mathbf{a},\mathbf{x}}(\mathbf{a}',\mathbf{x}')$ is an empty product. Thus $H_{\mathbf{a},\mathbf{x}}(\mathbf{a}',\mathbf{x}')=1/(2m)^n$ for every $\mathbf{a}'$, which implies that $\max_{\mathbf{a}'}H_{\mathbf{a},\mathbf{x}}(\mathbf{a}',\mathbf{x}')=\min_{\mathbf{a}'}H_{\mathbf{a},\mathbf{x}}(\mathbf{a}',\mathbf{x}')$. So for this case $\Delta_{\mathbf{a},\mathbf{x}}(\mathbf{x}')=0$.

Now consider the case when $K(\mathbf{x},\mathbf{x}')\neq\varnothing$, i.e., $k(\mathbf{x},\mathbf{x}')\ge 1$. For each $i\in K(\mathbf{x},\mathbf{x}')$, the factor $1+m(-1)^{a_i\oplus a'_i}$ has only two possible values: If $a_i=a'_i$, then $a_i\oplus a'_i=0$, so $1+m(-1)^{a_i\oplus a'_i}=1+m$. And if $a_i\neq a'_i$, then $a_i\oplus a'_i=1$, so $1+m(-1)^{a_i\oplus a'_i}=1-m$. Therefore each factor in the product $\prod_{i\in K(\mathbf{x},\mathbf{x}')}(1+m(-1)^{a_i\oplus a'_i})$ is either $1+m$ or $1-m$. Now suppose that among the $k$ indices in $K(\mathbf{x},\mathbf{x}')$, exactly $r$ of them satisfy $a_i\neq a'_i$. Then exactly $k-r$ of them satisfy $a_i=a'_i$, and the product becomes $(1+m)^{k-r}(1-m)^r$, or equivalently, $(-1)^r(m+1)^{k-r}(m-1)^r$. Thus the possible values of the product is 
\begin{equation*}
P_r = (-1)^r(m+1)^{k-r}(m-1)^r,\,\text{for }r=0,1,\ldots,k.
\end{equation*}
$P_r$ is positive when $r$ is even, and it is negative when $r$ is odd. The largest possible value occurs at $r=0$. So $\max_{r}P_r=P_0=(m+1)^k$. This maximum is achieved when $a_i=a'_i$ for all $i\in K(\mathbf{x},\mathbf{x}')$. Therefore,
\begin{equation}
\max_{\mathbf{a}'}H_{\mathbf{a},\mathbf{x}}(\mathbf{a}',\mathbf{x}') = \frac{(m+1)^k}{(2m)^n}.
\end{equation}
To find the minimum, notice that it must occur among the negative values, so $r$ must be odd. The \textit{magnitude} among odd $r$ is
$(m+1)^{k-r}(m-1)^r$. This is largest when $r=1$, because each time $r$ increases by $2$, the magnitude is multiplied by $\left(\frac{m-1}{m+1}\right)^2<1$. Hence, the most negative product occurs at $r=1$, i.e.,
\begin{equation*}
P_1 = (m+1)^{k-1}(1-m) = -(m-1)(m+1)^{k-1}.
\end{equation*}
Therefore,
\begin{equation}
\min_{\mathbf{a}'}H_{\mathbf{a},\mathbf{x}}(\mathbf{a}',\mathbf{x}') = -\frac{(m-1)(m+1)^{k-1}}{(2m)^n},
\end{equation}
which is achieved when exactly one index $i\in K(\mathbf{x},\mathbf{x}')$ has $a_i\neq a'_i$, and all other matched indices satisfy $a_j=a'_j$. And so, per the definition of $\Delta_{\mathbf{a},\mathbf{x}}(\mathbf{x}')$, subtracting the minimum value from the maximum value followed by some algebraic simplification gives:
\begin{equation}
\Delta_{\mathbf{a},\mathbf{x}}(\mathbf{x}') = \frac{(m+1)^k}{(2m)^n} - \left(-\frac{(m-1)(m+1)^{k-1}}{(2m)^n}\right) = \frac{2m(m+1)^{k-1}}{(2m)^n}.
\end{equation}
Finally, tying all the results together, we have:
\begin{align*}
\Delta_{\mathbf{a},\mathbf{x}}(\mathbf{x}') = \begin{cases}
0, &k(\mathbf{x},\mathbf{x}')=0,\\
\frac{2m(m+1)^{k(\mathbf{x},\mathbf{x}')-1}}{(2m)^n}, &k(\mathbf{x},\mathbf{x}')\ge 1.
\end{cases}
\end{align*}
To obtain the probability bound for the balanced case in~\eqref{eq:ProbBound_BalancedCase} when $N_{\mathbf{x}'}=N$, we just need to derive an explicit expression for the term $\sum_{\mathbf{x}'}\Delta_{\mathbf{a},\mathbf{x}}(\mathbf{x}')^2/N_{\mathbf{x}'}$ appearing in the denominator of~\eqref{eq:ProbBound_UnbalancedCase}. We have in the balanced case for $k(\mathbf{x},\mathbf{x}')\ge 1$
\begin{equation*}
\sum_{\mathbf{x}'}\frac{\Delta_{\mathbf{a},\mathbf{x}}(\mathbf{x}')^2}{N_{\mathbf{x}'}} = \frac{4m^2}{N(2m)^{2n}}\sum_{\mathbf{x}'}(m+1)^{2(k(\mathbf{x},\mathbf{x}')-1)}.
\end{equation*}
For fixed $\mathbf{x}$, the number of $\mathbf{x}'$ such that $k(\mathbf{x},\mathbf{x}')=k$ is $\binom{n}{k}(m-1)^{n-k}$. This is because if there are $k$ positions such that $x_i=x'_i$, then we can choose any of the $m-1$ settings different from $x_i$ for the remaining $n-k$ positions. Therefore,
\begin{align*}
\sum_{\mathbf{x}'}\frac{\Delta_{\mathbf{a},\mathbf{x}}(\mathbf{x}')^2}{N_{\mathbf{x}'}} &= \frac{4m^2}{N(2m)^{2n}}\sum_{k=1}^n\binom{n}{k}(m-1)^{n-k}(m+1)^{2(k-1)}\\
&= \frac{4m^2}{N(2m)^{2n}(m+1)^2}\sum_{k=1}^n\binom{n}{k}(m-1)^{n-k}(m+1)^{2k}\\
&= \frac{4m^2}{N(2m)^{2n}(m+1)^2}[((m-1)+(m+1)^2)^n-(m-1)^n]\\
&= \frac{4m^2}{N(2m)^{2n}(m+1)^2}[m^n(m+3)^n-(m-1)^n].
\end{align*}
Plugging this expression with $\gamma\coloneqq\min_{\mathbf{a},\mathbf{x}}p_0(\mathbf{a}|\mathbf{x})$ in~\eqref{eq:ProbBound_UnbalancedCase} gives the expression in~\eqref{eq:ProbBound_BalancedCase}, where $\sum_{\mathbf{x}'}\Delta_{\mathbf{a},\mathbf{x}}(\mathbf{x}')^2/N_{\mathbf{x}'}=C_{n,m}/N$ with $C_{n,m}=\frac{4m^2}{(2m)^{2n}(m+1)^2}[m^n(m+3)^n-(m-1)^n]$.

\section{Proof of Proposition~\ref{prop:CorrVecBelongsToNullSpace}}\label{a:ProofProp1}
\begin{prop*}
For all $\tilde{\mathbf{u}}_{I}\in\mathcal{S}^{k}$, the UMC coefficient vector $\mathbf{c}_{\tilde{\mathbf{u}}_{I}}$ satisfies $\mathbf{c}_{\tilde{\mathbf{u}}_{I}}\in\mathrm{ker}(A_{\mathrm{eq}})$.    
\end{prop*}

Before proceeding with the proof, we recap the coefficient vectors encoding the no-signalling and normalisation conditions, and the uniformly-averaged marginal correlators. To aid readability, we provide some illustrations of the coefficient vectors from the $(2,2,2)$ Bell scenario in Tables~\ref{tab:NScondition222NSB},~\ref{tab:Norm_condition222NSB}, and~\ref{tab:CorrCoeffVec222}.

For any fixed $(\tilde{\mathbf{a}}_{\neg i},\tilde{\mathbf{x}}_{\neg i})\in\mathcal{O}^{n-1}\times\mathcal{S}^{n-1}$, and for all $i\in[n],r\in[m]\setminus\{1\},(\mathbf{a},\mathbf{x})\in\mathcal{O}^n\times\mathcal{S}^n$, the no-signalling equality
$p(\tilde{\mathbf{a}}_{\neg i},\tilde{\mathbf{x}}_{\neg i},x_i = \mathsf{s}_1) - p(\tilde{\mathbf{a}}_{\neg i},\tilde{\mathbf{x}}_{\neg i},x_i = \mathsf{s}_r) = 0$
can be written as $\big(\mathbf{ns}_{(\tilde{\mathbf{a}}_{\neg i},\tilde{\mathbf{x}}_{\neg i})}^{i,r}\big)^T\mathbf{p}=0$ with $\mathbf{ns}_{(\tilde{\mathbf{a}}_{\neg i},\tilde{\mathbf{x}}_{\neg i})}^{i,r}$ defined componentwise as
\begin{equation*}
\mathrm{ns}_{(\tilde{\mathbf{a}}_{\neg i},\tilde{\mathbf{x}}_{\neg i})}^{i,r}(\mathbf{a},\mathbf{x})\coloneqq\delta_{\mathbf{a}_{\neg i},\tilde{\mathbf{a}}_{\neg i}}\delta_{\mathbf{x}_{\neg i},\tilde{\mathbf{x}}_{\neg i}}(\delta_{x_i,\mathsf{s}_1} - \delta_{x_i,\mathsf{s}_r}).
\end{equation*}
For any fixed $\tilde{\mathbf{x}}\in\mathcal{S}^n$, the normalization condition $\sum_{a}p(\mathbf{a}\lvert\tilde{\mathbf{x}})=1$ can be represented as $\big(\mathbf{nrm}_{\tilde{\mathbf{x}}}\big)^T\mathbf{p}=1$, where $\mathbf{nrm}_{\tilde{\mathbf{x}}}$ is defined componentwise as
\begin{equation*}
\mathrm{nrm}_{\tilde{\mathbf{x}}}(\mathbf{a},\mathbf{x})\coloneqq\delta_{\mathbf{x},\tilde{\mathbf{x}}}.
\end{equation*}
Next, for a choice of $k$-party index set $I\coloneqq\{i_1,i_2,\ldots,i_k\}\subseteq[n]$, the uniformly-averaged marginal correlator coefficient vector $\mathbf{c}_{\tilde{\mathbf u}_I}^I$ is defined componentwise as
\begin{equation*}
c_{\tilde{\mathbf u}_I}^I(\mathbf{a},\mathbf{x})\coloneqq\frac{1}{m^{n-\abs{I}}}\chi_{I}(\mathbf a)\delta_{\mathbf{x}_I,\tilde{\mathbf u}_I},\text{ where }\chi_I(\mathbf{a})=(-1)^{\bigoplus_{j\in I}a_j}
\end{equation*}

\begin{proof}
We prove the result by showing that the dot product of any row  of $A_{\mathrm{eq}}$ and the $k$-party correlator coefficient vector is zero. Any row of $A_{\mathrm{eq}}$ is either a no-signalling condition coefficient vector or a normalisation condition coefficient vector.

\textit{Dot product of $\mathbf{c}_{\tilde{\mathbf{u}}_{I}}$ and $\mathbf{ns}_{\tilde{\mathbf{a}}_{\neg i},\tilde{\mathbf{x}}_{\neg i}}^{i,r}$ is zero}: We first prove the result for $I=[n]$, i.e., we first prove it for a full-correlator coefficient vector $\mathbf{c}_{\tilde{\mathbf{u}}}$. The following holds for any fixed choice of setting configuration $\tilde{\mathbf{u}}\in\mathcal{S}^{n}$ (characterising the particular full correlator under consideration), and any fixed choices of party $i\in[n]$, non-$i$-party settings/outcomes $(\tilde{\mathbf{a}}_{\neg i},\tilde{\mathbf{x}}_{\neg i})\in\mathcal{O}^{n-1}\times\mathcal{S}^{n-1}$, and alternate setting choice $r\in[m]\setminus\{1\}$ for party $i$ (characterising the no-signalling condition):
\begin{align}
&\big(\mathbf{ns}^{i,r}_{(\tilde{\mathbf{a}}_{\neg i},\tilde{\mathbf{x}}_{\neg i})}\big)^{T}\mathbf{c}_{\tilde{\mathbf{u}}}\label{eq:InnerProdFullCorrNSCoeffVecZero}\\
&=\sum_{\mathbf{a},\mathbf{x}}\mathrm{ns}^{i,r}_{(\tilde{\mathbf{a}}_{\neg i},\tilde{\mathbf{x}}_{\neg i})}(\mathbf{a},\mathbf{x})\mathrm{c}_{\tilde{\mathbf{u}}}(\mathbf{a},\mathbf{x})\nonumber\\
&= \sum_{\mathbf{a},\mathbf{x}}\delta_{\mathbf{a}_{\neg i},\tilde{\mathbf{a}}_{\neg i}}\delta_{\mathbf{x}_{\neg i},\tilde{\mathbf{x}}_{\neg i}}(\delta_{x_{i},\mathsf{s}_{1}}-\delta_{x_{i},\mathsf{s}_{r}})(-1)^{\bigoplus_{j=1}^{n}a_{j}}\delta_{\mathbf{x},\tilde{\mathbf{u}}}\nonumber\\
&=\sum_{\mathbf{a}_{\neg i},\mathbf{x}_{\neg i}}\sum_{a_{i},x_{i}}\delta_{\mathbf{a}_{\neg i},\tilde{\mathbf{a}}_{\neg i}}\delta_{\mathbf{x}_{\neg i},\tilde{\mathbf{x}}_{\neg i}}\delta_{\mathbf{x}_{\neg i},\tilde{\mathbf{u}}_{\neg i}}(\delta_{x_{i},\mathsf{s}_{1}}-\delta_{x_{i},\mathsf{s}_{r}})(-1)^{\bigoplus_{j=1}^{n}a_{j}}\delta_{x_{i},\tilde{u}_{i}}\nonumber
\end{align}
In the double summation following the last equality in~\eqref{eq:InnerProdFullCorrNSCoeffVecZero}, the outer sum runs over all possible combinations of outcomes $\mathbf{a}_{\neg i}$ and inputs $\mathbf{x}_{\neg i}$, and the summands are immediately seen to be zero except when $\mathbf{a}_{\neg i},\mathbf{x}_{\neg i}$ are equal to the particular outcome-input combination $(\tilde{\mathbf{a}}_{\neg i},\tilde{\mathbf{x}}_{\neg i})$ corresponding to the vector $\mathbf{ns}^{i,r}_{(\tilde{\mathbf{a}}_{\neg i},\tilde{\mathbf{x}}_{\neg i})}$. Similarly, the particular input combination $\tilde{\mathbf{u}}_{\neg i}$ corresponding to the full-correlator coefficient vector should be equal to $\tilde{\mathbf{x}}_{\neg i}$ for the summand to not trivially reduce to zero. The summation then reduces in the restricted remaining case to the following:
\begin{align}
&\sum_{a_{i},x_{i}}(\delta_{x_{i},\mathsf{s}_{1}}-\delta_{x_{i},\mathsf{s}_{r}})(-1)^{\bigoplus_{j\neq i}\tilde{a}_{j}\oplus a_{i}}\delta_{x_{i},\tilde{u}_{i}}\nonumber\\
=& \sum_{x_{i}}(\delta_{x_{i},\mathsf{s}_{1}}-\delta_{x_{i},\mathsf{s}_{r}})\delta_{x_{i},\tilde{u}_{i}}\big\{(-1)^{\bigoplus_{j\neq i}\tilde{a}_{j}\oplus 0} + (-1)^{\bigoplus_{j\neq i}\tilde{a}_{j}\oplus 1}\big\}=0.\label{e:plusandminusone}
\end{align}
The last equality follows from the fact that if $\bigoplus_{j\neq i}\tilde{a}_{j}=0$, then $\bigoplus_{j\neq i}\tilde{a}_{j}\oplus 0=0$ and $\bigoplus_{j\neq i}\tilde{a}_{j}\oplus 1=1$, which results in a cancellation of terms and hence a zero. The same holds if $\bigoplus_{j\neq i}\tilde{a}_{j}=1$.

Next, we consider the case where $|I|<n$, i.e., we consider a $k$-party correlator coefficient vector $\mathbf{c}_{\tilde{\mathbf{u}}_{I}}$ for $k<n$. The following holds for any fixed choice of $i\in[n]$, $r\in[m]\setminus\{1\}$, $\tilde{\mathbf{u}}_{I}\in\mathcal{S}^k$, and $(\tilde{\mathbf{a}}_{\neg i},\tilde{\mathbf{x}}_{\neg i})$,:
\begin{align}\label{eq:InnerProdCorrNSVec_Zero}
&\big(\mathbf{ns}^{i,r}_{(\tilde{\mathbf{a}}_{\neg i},\tilde{\mathbf{x}}_{\neg i})}\big)^T\mathbf{c}_{\tilde{\mathbf{u}}_{I}}\nonumber\\
&= \sum_{\mathbf{a},\mathbf{x}}\mathrm{ns}^{i,r}_{\tilde{\mathbf{a}}_{\neg i},\tilde{\mathbf{x}}_{\neg i}}(\mathbf{a},\mathbf{x})\mathrm{c}_{\tilde{\mathbf{u}}_{I}}(\mathbf{a},\mathbf{x})\nonumber\\
    &= \sum_{\mathbf{a}_{\neg i},\mathbf{x}_{\neg i}}\sum_{a_{i},x_{i}}\delta_{\mathbf{a}_{\neg i},\tilde{\mathbf{a}}_{\neg i}}\delta_{\mathbf{x}_{\neg i},\tilde{\mathbf{x}}_{\neg i}}\left(\delta_{x_{i},\mathsf{s}_{1}}-\delta_{x_{i},\mathsf{s}_{r}}\right)\frac{1}{m^{n-|I|}}(-1)^{\bigoplus_{j\in I}a_{j}}\delta_{\mathbf{x}_{I},\tilde{\mathbf{u}}_{I}}
\end{align}
We can analyse the above according to the following two cases. 
\begin{itemize}
\item $i \in I$: This means that the party $i$ singled out by the no-signalling condition is one of the $k$ parties of interest in the UMC coefficient vector $\mathbf{c}_{\tilde{\mathbf{u}}_{I}}$. In this case, the summands in~\eqref{eq:InnerProdCorrNSVec_Zero} immediately reduce to zero unless $\mathbf{a}_{\neg i}=\tilde{\mathbf{a}}_{\neg i}$ and $\mathbf{x}_{\neg i}=\tilde{\mathbf{x}}_{\neg i}$; when we fix this choice for the outer sum, the expression reduces to
\begin{equation}\label{eq:InnerProdCorrNSVec_Zero_1}
\frac{1}{m^{n-|I|}}\sum_{x_i}(\delta_{x_i,\mathsf{s}_1}-\delta_{x_i,\mathsf{s}_r})\delta_{x_i,\tilde{u}_i}\sum_{a_i}(-1)^{a_i \oplus \big(\bigoplus_{j\in I \setminus \{i\}}\tilde{a}_j\big)}
\end{equation}
Notice that fixing the choice $(\mathbf{a}_{\neg i},\mathbf{x}_{\neg i})=(\tilde{\mathbf{a}}_{\neg i},\tilde{\mathbf{x}}_{\neg i})$ will still result in \eqref{eq:InnerProdCorrNSVec_Zero} immediately taking the value of zero unless the corresponding entries in $\tilde{\mathbf x}_{\neg i}$ and $\tilde{\mathbf u}_{I}$ (except the $i$'th entry) are the same, which is the scenario that ultimately results in~\eqref{eq:InnerProdCorrNSVec_Zero_1}. The above expression reduces to zero once we observe the following:
\begin{equation*}
    \sum_{a_i}(-1)^{a_i \oplus \big(\bigoplus_{j\in I \setminus \{i\}}\tilde{a}_j\big)} = (+1) + (-1) = 0,
\end{equation*}
\item $i\notin I$: In this case, looking back again to the indices of the outer sum in~\eqref{eq:InnerProdCorrNSVec_Zero}, the summands continue to immediately reduce to zero unless $(\mathbf{a}_{\neg i},\mathbf{x}_{\neg i})=(\tilde{\mathbf{a}}_{\neg i},\tilde{\mathbf{x}}_{\neg i})$ hold, but now additionally the final delta function in the expression depends entirely on $\mathbf{x}_{\neg i}$ as well (as $i$ is now not an element of $I$), such that it evaluates to zero if we do not additionally have agreement of $\tilde{u}_{I}$ with $\mathbf{x}_{\neg i}$. When all the above non-zero conditions do hold, the expression in~\eqref{eq:InnerProdCorrNSVec_Zero} reduces to
\begin{equation*}
\frac{1}{m^{n-|I|}}\sum_{a_i}(-1)^{\bigoplus_{j\in I}\tilde{a}_j}\sum_{x_i}(\delta_{x_i,\mathsf{s}_1}-\delta_{x_i,\mathsf{s}_r})=0. 
\end{equation*}
\end{itemize}

\medskip

\noindent \textit{Dot product of $\mathbf{c}_{\tilde{\mathbf{u}}_{I}}$ and $\mathbf{nrm}_{\tilde{\mathbf{x}}}$ is zero}: As before, we first show that it holds for the case $I=[n]$, i.e., for the full-correlator coefficient vector $\mathbf{c}_{\tilde{\mathbf{u}}}$.
\begin{equation*}
\big(\mathbf{nrm}_{\tilde{\mathbf{x}}}\big)^T\mathbf{c}_{\tilde{\mathbf{u}}}= \sum_{\mathbf{a},\mathbf{x}}\delta_{\mathbf{x},\tilde{\mathbf{x}}}(-1)^{\bigoplus_{j=1}^n a_{j}}\delta_{\mathbf{x},\tilde{\mathbf{u}}}=\delta_{\tilde{\mathbf{x}},\tilde{\mathbf{u}}}\sum_{\mathbf{a}}(-1)^{\bigoplus_{j=1}^n a_{j}}=0.
\end{equation*}
The final equality follows from the fact that the contribution to the sum from the even-parity $n$-bit strings is $2^{n-1}$ and that from the odd-parity $n$-bit strings is $-2^{n-1}$, which cancel out.

Next, we consider the case where $|I|<n$. 
\begin{align}
\big(\mathbf{nrm}_{\tilde{\mathbf{x}}}\big)^T\mathbf{c}_{\tilde{\mathbf{u}}_{I}}
&=\sum_{\mathbf{a},\mathbf{x}}\delta_{\mathbf{x},\tilde{\mathbf{x}}}\frac{1}{m^{n-|I|}}(-1)^{\bigoplus_{j\in I}a_{j}}\delta_{\mathbf{x}_{I},\tilde{\mathbf{u}}_{I}}\label{eq:InnerProdNormVecCorrCoeffVecZero}\\
&=\sum_{\mathbf{x}_{I}}\sum_{\mathbf{x}_{I}}\sum_{\mathbf{a}}\delta_{\mathbf{x}_{I},\tilde{\mathbf{x}}_{I}}\delta_{\mathbf{x}_{I},\tilde{\mathbf{x}}_{I}}\delta_{\mathbf{x}_{I},\tilde{\mathbf{u}}_{I}}\frac{1}{m^{n-|I|}}(-1)^{\bigoplus_{j\in I}a_{j}}\nonumber \\
&=\delta_{\tilde{\mathbf{x}}_{I},\tilde{\mathbf{u}}_{I}}\sum_{\mathbf{x}_{I}}\delta_{\mathbf{x}_{I},\tilde{\mathbf{x}}_{I}}\sum_{\mathbf{a}}\frac{1}{m^{n-|I|}}(-1)^{\bigoplus_{j\in I}a_{j}}.\nonumber
\end{align}
Regardless of the value of the delta functions, the above expression will always equal zero  because
\begin{equation*}\sum_{\mathbf{a}}\frac{1}{m^{n-|I|}}(-1)^{\bigoplus_{j\in I}a_{j}}=\sum_{\mathbf{a}_{I}}(-1)^{\bigoplus_{j\in I}a_{j}}\sum_{\mathbf{a}_{I}}\frac{1}{m^{n-|I|}}=\frac{2^{n-|I|}}{m^{n-|I|}}\sum_{\mathbf{a}_{I}}(-1)^{\bigoplus_{j\in I}a_{j}}=0\nonumber.
\end{equation*}
\end{proof}

\begin{table}[H]
\centering
\renewcommand{\arraystretch}{1.1}
\caption{The no-signalling condition $p(0,0\lvert\mathsf{a},\mathsf{b}) + p(1,0\lvert\mathsf{a},\mathsf{b}) = p(0,0\lvert\mathsf{a}',\mathsf{b}) + p(1,0\lvert\mathsf{a}',\mathsf{b})$ for the $(2,2,2)$ Bell scenario is expressible as $(\mathbf{ns}_{(0,\mathsf{b})}^{1,\mathsf{a}'})^T\mathbf{p}=0$, where $\mathbf{ns}_{(0,\mathsf{b})}^{1,\mathsf{a}'}$ is shown below (the subscript corresponds to $(b_1,y_1)=(0,\mathsf{b})$ and the superscripts $1$ and $\mathsf{a}'$ correspond to party $\mathsf{A}_1$ and the setting-choice $\mathsf{a}'$; these are the identifiers of a particular no-signalling condition). We arrange the components $\mathrm{ns}_{(0,\mathsf{b})}^{1,\mathsf{a}'}(a_1,a_2,x_1,x_2)$ and $p(a_1,a_2\lvert x_1,x_2)$ in lexicographic order.}
\resizebox{\linewidth}{!}{%
\begin{tabular}{cccccccccccccccc}
\hline\hline
$00\mathsf{a}\mathsf{b}$ & $00\mathsf{a}\mathsf{b}'$ & $00\mathsf{a}'\mathsf{b}$ & $00\mathsf{a}'\mathsf{b}'$ & $01\mathsf{a}\mathsf{b}$ & $01\mathsf{a}\mathsf{b}'$ & $01\mathsf{a}'\mathsf{b}$ & $01\mathsf{a}'\mathsf{b}'$ & $10\mathsf{a}\mathsf{b}$ & $10\mathsf{a}\mathsf{b}'$ & $10\mathsf{a}'\mathsf{b}$ & $10\mathsf{a}'\mathsf{b}'$ & $11\mathsf{a}\mathsf{b}$ & $11\mathsf{a}\mathsf{b}'$ & $11\mathsf{a}'\mathsf{b}$ & $11\mathsf{a}'\mathsf{b}'$ \\
\hline
1 & 0 & -1 & 0 & 0 & 0 & 0 & 0 & 1 & 0 & -1 & 0 & 0 & 0 & 0 & 0 \\
\hline\hline
\end{tabular}}
\label{tab:NScondition222NSB}
\end{table}

\begin{table}[H]
\centering
\caption{\label{tab:Norm_condition222NSB}The normalisation condition $p(0,0\lvert\mathsf{a},\mathsf{b}) + p(0,1\lvert\mathsf{a},\mathsf{b}) + p(1,0\lvert\mathsf{a},\mathsf{b}) + p(1,1\lvert\mathsf{a},\mathsf{b}) = 1$ for the $(2,2,2)$ Bell scenario is expressible as $\mathbf{nrm}_{\mathsf{a},\mathsf{b}}^T\mathbf{p}=1$ with $\mathbf{nrm}_{\mathsf{a},\mathsf{b}}$ as shown below. The subscripts of $\mathbf{nrm}_{\mathsf{a},\mathsf{b}}$ correspond to the settings combination $(x_1,x_2)=(\mathsf{a},\mathsf{b})$ identifying this particular normalisation condition.}
\renewcommand{\arraystretch}{1.1}
\resizebox{\linewidth}{!}{%
\begin{tabular}{cccccccccccccccc}
\hline\hline
$00\mathsf{a}\mathsf{b}$ & $00\mathsf{a}\mathsf{b}'$ & $00\mathsf{a}'\mathsf{b}$ & $00\mathsf{a}'\mathsf{b}'$ & $01\mathsf{a}\mathsf{b}$ & $01\mathsf{a}\mathsf{b}'$ & $01\mathsf{a}'\mathsf{b}$ & $01\mathsf{a}'\mathsf{b}'$ & $10\mathsf{a}\mathsf{b}$ & $10\mathsf{a}\mathsf{b}'$ & $10\mathsf{a}'\mathsf{b}$ & $10\mathsf{a}'\mathsf{b}'$ & $11\mathsf{a}\mathsf{b}$ & $11\mathsf{a}\mathsf{b}'$ & $11\mathsf{a}'\mathsf{b}$ & $11\mathsf{a}'\mathsf{b}'$ \\
\hline
1 & 0 & 0 & 0 & 1 & 0 & 0 & 0 & 1 & 0 & 0 & 0 & 1 & 0 & 0 & 0 \\
\hline\hline
\end{tabular}}
\end{table}

\begin{table}[H]
\caption{\label{tab:CorrCoeffVec222}The eight UMC coefficient vectors for the $(2,2,2)$ Bell scenario: $\mathbf{c}_{\mathsf{a}}^{\{1\}}$, $\mathbf{c}_{\mathsf{a}'}^{\{1\}}$,$\mathbf{c}_{\mathsf{b}}^{\{2\}}$, $\mathbf{c}_{\mathsf{b}'}^{\{2\}}$, $\mathbf{c}_{\mathsf{a},\mathsf{b}}^{\{1,2\}}$, $\mathbf{c}_{\mathsf{a},\mathsf{b}'}^{\{1,2\}}$, $\mathbf{c}_{\mathsf{a}',\mathsf{b}}^{\{1,2\}}$ and $\mathbf{c}_{\mathsf{a}',\mathsf{b}'}^{\{1,2\}}$ defining the UMC terms $\bar{C}_{\tilde{\mathbf u}_I}^I$ as $(\mathbf{c}_{\tilde{\mathbf u}_I}^I)^T\mathbf{p}$, where $\varnothing\neq I\subseteq\{1,2\}$.}
\centering
\renewcommand{\arraystretch}{1.16}
\resizebox{\linewidth}{!}{%
\begin{tabular}{ccccccccccccccccc}
\hline\hline
 & $00\mathsf{a}\mathsf{b}$ & $00\mathsf{a}\mathsf{b}'$ & $00\mathsf{a}'\mathsf{b}$ & $00\mathsf{a}'\mathsf{b}'$ & $01\mathsf{a}\mathsf{b}$ & $01\mathsf{a}\mathsf{b}'$ & $01\mathsf{a}'\mathsf{b}$ & $01\mathsf{a}'\mathsf{b}'$ & $10\mathsf{a}\mathsf{b}$ & $10\mathsf{a}\mathsf{b}'$ & $10\mathsf{a}'\mathsf{b}$ & $10\mathsf{a}'\mathsf{b}'$ & $11\mathsf{a}\mathsf{b}$ & $11\mathsf{a}\mathsf{b}'$ & $11\mathsf{a}'\mathsf{b}$ & $11\mathsf{a}'\mathsf{b}'$ \\
\hline
$\mathbf{c}_{\mathsf{a}}^{\{1\}}$ & $\frac{1}{2}$ & $\frac{1}{2}$ & $0$ & $0$ & $\frac{1}{2}$ & $\frac{1}{2}$ & $0$ & $0$ & $-\frac{1}{2}$ & $-\frac{1}{2}$ & $0$ & $0$ & $-\frac{1}{2}$ & $-\frac{1}{2}$ & 0 & 0 \\
$\mathbf{c}_{\mathsf{a}'}^{\{1\}}$ & $0$ & $0$ & $\frac{1}{2}$ & $\frac{1}{2}$ & $0$ & $0$ & $\frac{1}{2}$ & $\frac{1}{2}$ & $0$ & $0$ & $-\frac{1}{2}$ & $-\frac{1}{2}$ & 0 & 0 & $-\frac{1}{2}$ & $-\frac{1}{2}$ \\
$\mathbf{c}_{\mathsf{b}}^{\{2\}}$ & $\frac{1}{2}$ & $0$ & $\frac{1}{2}$ & $0$ & $-\frac{1}{2}$ & $0$ & $-\frac{1}{2}$ & $0$ & $\frac{1}{2}$ & $0$ & $\frac{1}{2}$ & $0$ & $-\frac{1}{2}$ & $0$ & $-\frac{1}{2}$ & $0$ \\
$\mathbf{c}_{\mathsf{b}'}^{\{2\}}$ & $0$ & $\frac{1}{2}$ & $0$ & $\frac{1}{2}$ & $0$ & $-\frac{1}{2}$ & $0$ & $-\frac{1}{2}$ & $0$ & $\frac{1}{2}$ & $0$ & $\frac{1}{2}$ & $0$ & $-\frac{1}{2}$ & $0$ & $-\frac{1}{2}$ \\
$\mathbf{c}_{\mathsf{a},\mathsf{b}}^{\{1,2\}}$ & $1$ & $0$ & $0$ & $0$ & $-1$ & $0$ & $0$ & $0$ & $-1$ & $0$ & $0$ & $0$ & $1$ & $0$ & $0$ & $0$ \\
$\mathbf{c}_{\mathsf{a},\mathsf{b}'}^{\{1,2\}}$ & $0$ & $1$ & $0$ & $0$ & $0$ & $-1$ & $0$ & $0$ & $0$ & $-1$ & $0$ & $0$ & $0$ & $1$ & $0$ & $0$ \\
$\mathbf{c}_{\mathsf{a}',\mathsf{b}}^{\{1,2\}}$ & $0$ & $0$ & $1$ & $0$ & $0$ & $0$ & $-1$ & $0$ & $0$ & $0$ & $-1$ & $0$ & $0$ & $0$ & $1$ & $0$ \\
$\mathbf{c}_{\mathsf{a}',\mathsf{b}'}^{\{1,2\}}$ & $0$ & $0$ & $0$ & $1$ & $0$ & $0$ & $0$ & $-1$ & $0$ & $0$ & $0$ & $-1$ & $0$ & $0$ & $0$ & $1$ \\
\hline\hline
\end{tabular}}
\end{table}

\section{Projection in the \texorpdfstring{$(2,2,2)$}{222} Bell scenario}

The matrix representations of the linear maps $T_i$ in the projection $\widehat{\mathbf p} = T_3 T_2 T_1\mathbf{f}$ depend on the chosen ordering of the components of the empirical behaviour $\mathbf f$, the intermediate vectors, and the final $L^2$-estimate $\widehat{\mathbf p}$. By arranging these components in a suitable order, one can expose additional algebraic structure in the maps $T_i$, for example, block structure or Kronecker structure, and thereby express the $T_i$'s more compactly. Representing the individual conditional probabilities $f(a_1,a_2\lvert x_1,x_2)$ as $f_{a_{1}a_{2}\lvert x_{1}x_{2}}$, we choose the following ordering for the components of $\mathbf f$, which groups outcomes by setting configuration. The input choices $x_1,x_2$ and the possible outcomes $a_1,a_2$ for each input choice are from the set $\{0,1\}$.
\begin{equation*}
\mathbf{f}=\begin{pmatrix}
\mathbf{f}_{00} & \mathbf{f}_{01} & \mathbf{f}_{10} & \mathbf{f}_{11} 
\end{pmatrix}^T,\text{ where } \mathbf{f}_{x_1 x_2} = \begin{pmatrix}
f_{00\lvert x_1 x_2} & f_{01\lvert x_1 x_2} & f_{10\lvert x_1 x_2} & f_{11\lvert x_1 x_2}  
\end{pmatrix}^T.
\end{equation*}
(1) \textit{From probabilities to parity correlators, $\mathbf{p}\mapsto C_{\mathbf x}^I$}: For each input pair $(x_1,x_2)$ and subset $I\subseteq\{1,2\}$ (including $I=\varnothing$), the parity correlators $C_{x_1 x_2}^{I}$ are as defined in~\eqref{eq:f_to_Cx}. We collect the $C_{x_1 x_2}^I$ terms into the vector $\mathbf{t}_1$ in the following order: We arrange $\mathbf{t}_1$ as four $4\times 1$ blocks, each block corresponding to an input pair $(x_1,x_2)$, and within each block the correlators are arranged in the order $\{\varnothing,\{1\},\{2\},\{1,2\}\}$.

\begin{equation}\label{eq:ordering_Cx}
\mathbf{t}_1 = \begin{pmatrix}
\mathbf{C}_{00}^I & \mathbf{C}_{01}^I & \mathbf{C}_{10}^I & \mathbf{C}_{11}^I
\end{pmatrix}^{T},\text{ where }\mathbf{C}_{x_1 x_2}^I=\begin{pmatrix}
1 & C_{x_1 x_2}^{\{1\}} & C_{x_1 x_2}^{\{2\}} & C_{x_1 x_2}^{\{1,2\}}    \end{pmatrix}^T.
\end{equation}
The matrix representation of the linear map $T_1$ is the following:

\begin{equation}\label{eq:T_1_MatrixRep}
T_1 = \begin{pmatrix}
1 & 1 & 1 & 1 & 0 & 0 & 0 & 0 & 0 & 0 & 0 & 0 & 0 & 0 & 0 & 0 \\
1 & 1 & -1 & -1 & 0 & 0 & 0 & 0 & 0 & 0 & 0 & 0 & 0 & 0 & 0 & 0 \\
1 & -1 & 1 & -1 & 0 & 0 & 0 & 0 & 0 & 0 & 0 & 0 & 0 & 0 & 0 & 0 \\
1 & -1 & -1 & 1 & 0 & 0 & 0 & 0 & 0 & 0 & 0 & 0 & 0 & 0 & 0 & 0 \\
0 & 0 & 0 & 0 & 1 & 1 & 1 & 1 & 0 & 0 & 0 & 0 & 0 & 0 & 0 & 0 \\
0 & 0 & 0 & 0 & 1 & 1 & -1 & -1 & 0 & 0 & 0 & 0 & 0 & 0 & 0 & 0 \\
0 & 0 & 0 & 0 & 1 & -1 & 1 & -1 & 0 & 0 & 0 & 0 & 0 & 0 & 0 & 0 \\
0 & 0 & 0 & 0 & 1 & -1 & -1 & 1 & 0 & 0 & 0 & 0 & 0 & 0 & 0 & 0 \\
0 & 0 & 0 & 0 & 0 & 0 & 0 & 0 & 1 & 1 & 1 & 1 & 0 & 0 & 0 & 0 \\
0 & 0 & 0 & 0 & 0 & 0 & 0 & 0 & 1 & 1 & -1 & -1 & 0 & 0 & 0 & 0 \\
0 & 0 & 0 & 0 & 0 & 0 & 0 & 0 & 1 & -1 & 1 & -1 & 0 & 0 & 0 & 0 \\
0 & 0 & 0 & 0 & 0 & 0 & 0 & 0 & 1 & -1 & -1 & 1 & 0 & 0 & 0 & 0 \\
0 & 0 & 0 & 0 & 0 & 0 & 0 & 0 & 0 & 0 & 0 & 0 & 1 & 1 & 1 & 1 \\
0 & 0 & 0 & 0 & 0 & 0 & 0 & 0 & 0 & 0 & 0 & 0 & 1 & 1 & -1 & -1 \\ 
0 & 0 & 0 & 0 & 0 & 0 & 0 & 0 & 0 & 0 & 0 & 0 & 1 & -1 & 1 & -1 \\
0 & 0 & 0 & 0 & 0 & 0 & 0 & 0 & 0 & 0 & 0 & 0 & 1 & -1 & -1 & 1
\end{pmatrix}.
\end{equation}
With this particular ordering for $\mathbf f$ and $\mathbf{t}_1$, the map $T_1$ has a block structure, where the Walsh-Hadamard matrix $H_4\in\mathbb{R}^{4\times 4}$ along the diagonal is expressible as $H_{4}=SH^{\otimes 2}$ for $H=\begin{pmatrix}
1 & 1 \\ 1 & -1\end{pmatrix}$ and the permutation matrix $S\in\mathbb{R}^{4\times 4}$ associated with the transposition $(2\,\,3)$ that swaps rows $2$ and $3$. $T_1$ can then be expressed as
\begin{equation*}
    T_1 = I_4\otimes H_4,
\end{equation*}
where $I_4\in\mathbb{R}^{4\times 4}$ is the identity matrix. Here $I_4$ indexes the four input pairs $(x_1,x_2)$, while $H_4$ converts the four outcomes for each pair into $\begin{pmatrix}
\bar{C}^{\varnothing} & \bar{C}^{\{1\}} & \bar{C}^{\{2\}} & \bar{C}^{\{1,2\}}\end{pmatrix}^T$ in the order of~\eqref{eq:ordering_Cx}.

(2) \textit{Uniformly averaging over non-$k$-party settings, $C_{\mathbf x}^I\mapsto\bar{C}_{\mathbf{x}_I}^I$}: We choose the following order for the vector $\mathbf{t}_2$ collecting the UMC terms $\bar{C}_{\mathbf{x}_I}^I$ for all $I\subseteq\{1,2\}$.
\begin{equation}\label{eq:ordering_CxI}
\mathbf{t}_2 = \begin{pmatrix}1 & \bar{C}_{0}^{\{1\}} & \bar{C}_{1}^{\{1\}} & \bar{C}_{0}^{\{2\}} & \bar{C}_{1}^{\{2\}} & \bar{C}_{00}^{\{1,2\}} & \bar{C}_{01}^{\{1,2\}} & \bar{C}_{10}^{\{1,2\}} & \bar{C}_{11}^{\{1,2\}}
\end{pmatrix}^{T}.
\end{equation}
The linear map, shown in~\eqref{eq:Cx_to_CxI}, is then represented as
\begin{equation}\label{eq:T_2_MatrixRep}
T_2 = \begin{pmatrix}
\frac{1}{4} & 0 & 0 & 0 & \frac{1}{4} & 0 & 0 & 0 & \frac{1}{4} & 0 & 0 & 0 & \frac{1}{4} & 0 & 0 & 0 \\
 0 & \frac{1}{2} & 0 & 0 & 0 & \frac{1}{2} & 0 & 0 & 0 & 0 & 0 & 0 & 0 & 0 & 0 & 0 \\
0 & 0 & 0 & 0 & 0 & 0 & 0 & 0 & 0 & \frac{1}{2} & 0 & 0 & 0 & \frac{1}{2} & 0 & 0 \\
0 & 0 & \frac{1}{2} & 0 & 0 & 0 & 0 & 0 & 0 & 0 & \frac{1}{2} & 0 & 0 & 0 & 0 & 0 \\
0 & 0 & 0 & 0 & 0 & 0 & \frac{1}{2} & 0 & 0 & 0 & 0 & 0 & 0 & 0 & \frac{1}{2} & 0 \\
0 & 0 & 0 & 1 & 0 & 0 & 0 & 0 & 0 & 0 & 0 & 0 & 0 & 0 & 0 & 0 \\
0 & 0 & 0 & 0 & 0 & 0 & 0 & 1 & 0 & 0 & 0 & 0 & 0 & 0 & 0 & 0 \\
0 & 0 & 0 & 0 & 0 & 0 & 0 & 0 & 0 & 0 & 0 & 1 & 0 & 0 & 0 & 0 \\
0 & 0 & 0 & 0 & 0 & 0 & 0 & 0 & 0 & 0 & 0 & 0 & 0 & 0 & 0 & 1
\end{pmatrix}.
\end{equation}
$T_2$ carries out four types of operations: (1) a $1/4$ average of the four ``1'' entries (one per $(x_1,x_2)$), (2) a $1/2$ average over Bob's input for $\bar{C}_{x_1 x_2}^{\{1\}}$, (3) a $1/2$ average over Alice's input for $\bar{C}_{x_1 x_2}^{\{2\}}$, and (4) identity on the four two-party (full) correlators $\bar{C}_{x_1 x_2}^{\{1,2\}}$. To make the structure of $T_2$ in terms of these operations explicit, it is convenient to re-arrange the components of $\mathbf{t}_1$ so that across all $(x_1,x_2)$, the components are grouped by the type $I\in\{\varnothing,\{1\},\{2\},\{1,2\}\}$. Let $P\in\mathbb{R}^{16\times 16}$ be the permutation matrix that when applied to $\mathbf{t}_1$ results in the re-ordering $[\varnothing\,|\,\{1\}\,|\,\{2\}\,|\,\{1,2\}]$ as shown below:
{\small\begin{equation*}
P\mathbf{t}_1 = \begin{pmatrix}
1 & 1 & 1 & 1 & C_{00}^{\{1\}} & C_{01}^{\{1\}} & C_{10}^{\{1\}} & C_{11}^{\{1\}} & C_{00}^{\{2\}} & C_{01}^{\{2\}} & C_{10}^{\{2\}} & C_{11}^{\{2\}} & C_{00}^{\{1,2\}} & C_{01}^{\{1,2\}} & C_{10}^{\{1,2\}} & C_{11}^{\{1,2\}}
\end{pmatrix}^{T}
\end{equation*}}
Denote the row vectors $\mathbf{r}_{\mathrm{avg}}=\begin{pmatrix}\frac{1}{2} & \frac{1}{2}\end{pmatrix}$ and $\mathbf{1}_4=\begin{pmatrix}
1 & 1 & 1 & 1\end{pmatrix}$. The matrix shown in~\eqref{eq:T_2_MatrixRep} can be expressed as shown below:
\begin{equation*}
T_2 = \left[\frac{1}{4}\mathbf{1}_4\oplus(I_2\otimes\mathbf{r}_{\mathrm{avg}})\oplus(\mathbf{r}_{\mathrm{avg}}\otimes I_2)\oplus I_4\right]P
\end{equation*}
This shows $T_2$ as a direct sum of simple averaging maps acting on the four grouped blocks.

(3) \textit{Reconstruction, $\bar{C}_{\mathbf{x}_I}^I\mapsto\widehat{\mathbf p}$}: Finally, we map $\mathbf{t}_2$ back to a behaviour $\widehat{\mathbf p}\in\mathbb{R}^{16}$ using the same component ordering as $\mathbf f$. The matrix representation of $T_3$, shown in~\eqref{eq:CxI_to_phat}, is 
\begin{equation}
T_3 = \frac{1}{4}\begin{pmatrix}
1 & 1 & 0 & 1 & 0 & 1 & 0 & 0 & 0 \\
1 & 1 & 0 & -1 & 0 & -1 & 0 & 0 & 0 \\
1 & -1 & 0 & 1 & 0 & -1 & 0 & 0 & 0 \\
1 & -1 & 0 & -1 & 0 & 1 & 0 & 0 & 0 \\
1 & 1 & 0 & 0 & 1 & 0 & 1 & 0 & 0\\
1 & 1 & 0 & 0 & -1 & 0 & -1 & 0 & 0\\
1 & -1 & 0 & 0 & 1 & 0 & -1 & 0 & 0\\
1 & -1 & 0 & 0 & -1 & 0 & 1 & 0 & 0\\
1 & 0 & 1 & 1 & 0 & 0 & 0 & 1 & 0\\
1 & 0 & 1 & -1 & 0 & 0 & 0 & -1 & 0\\
1 & 0 & -1 & 1 & 0 & 0 & 0 & -1 & 0\\
1 & 0 & -1 & -1 & 0 & 0 & 0 & 1 & 0\\
1 & 0 & 1 & 0 & 1 & 0 & 0 & 0 & 1\\
1 & 0 & 1 & 0 & -1 & 0 & 0 & 0 & -1\\
1 & 0 & -1 & 0 & 1 & 0 & 0 & 0 & -1\\
1 & 0 & -1 & 0 & -1 & 0 & 0 & 0 & 1\\
\end{pmatrix}    
\end{equation}
The columns of $T_3$ correspond to the correlators $(\bar{C}^{\varnothing},\bar{C}_{0}^{\{1\}},\bar{C}_{1}^{\{1\}},\bar{C}_{0}^{\{2\}},\bar{C}_{1}^{\{2\}},\bar{C}_{00}^{\{1,2\}},\bar{C}_{01}^{\{1,2\}},\bar{C}_{10}^{\{1,2\}},\bar{C}_{11}^{\{1,2\}})$ in the order as listed. The matrix can be seen as four matrices $\tilde{R}_{x_1 x_2}\in\mathbb{R}^{4\times 9}$, one for each input pair $(x_1,x_2)$, stacked vertically in the descending order $(\tilde{R}_{00},\tilde{R}_{01},\tilde{R}_{10},\tilde{R}_{11})$. The matrix $\tilde{R}_{x_1 x_2}$ can be expressed as
\begin{equation*}
\tilde{R}_{x_1 x_2} = H_{4}R_{x_1 x_2},\text{ where }R_{x_1 x_2} = \begin{pmatrix}\mathbf{e}_{1}^T \\ \mathbf{e}_{2+x_1}^T \\ \mathbf{e}_{4+x_2}^T \\ \mathbf{e}_{6+2x_1 + x_2}^T
\end{pmatrix} \in \mathbb{R}^{4\times 9}\text{ for }x_1,x_2\in\{0,1\}.
\end{equation*}
The vectors $\{\mathbf{e}_{j}\}_{j=1}^{9}\in\mathbb{R}^9$ are the standard basis vectors. The matrix $T_3$ is then expressible as
\begin{equation*}
T_3 = \frac{1}{4}\begin{pmatrix}\tilde{R}_{00} \\ \tilde{R}_{01} \\ \tilde{R}_{10} \\ \tilde{R}_{11}\end{pmatrix} =\frac{1}{4}(I_4\otimes H_4)\begin{pmatrix}R_{00}\\R_{01}\\R_{10}\\R_{11}\end{pmatrix}.
\end{equation*}

\printbibliography

\end{document}